\documentclass[12 pt]{article}
\usepackage{lmodern} 
\usepackage[T1]{fontenc}
\usepackage{textcomp}
\usepackage{multirow}
\usepackage{tabularx} 
\newcolumntype{C}{>{\centering\arraybackslash}X} 
\usepackage[scr=boondoxo,frak=boondox,bb=boondox]{mathalfa}
\usepackage{amsmath}             
\allowdisplaybreaks
\usepackage{amssymb}
\usepackage[mathcal]{euscript}
\usepackage{bm}
\usepackage{amsthm}
 
\usepackage{ragged2e}
\usepackage{breqn}
\usepackage{graphicx}
\usepackage{epstopdf}
\usepackage[inline]{enumitem}
\usepackage{tabto} 
\usepackage{multicol} 
\def\namedlabel#1#2{\begingroup
    #2%
    \def\@currentlabel{#2}%
    \phantomsection\label{#1}\endgroup
}
\usepackage{boxhandler}
\usepackage[title]{appendix}
\usepackage{lipsum}
\usepackage{authblk}
\usepackage[top=2cm, bottom=2cm, left=2cm, right=2cm]{geometry}
\usepackage{fancyhdr}

\usepackage{appendix}
\usepackage{hyperref}

\usepackage{enumitem}
\usepackage{float}
\usepackage{verbatim}
\usepackage{booktabs}
 \usepackage{setspace}
\renewenvironment{abstract}{
\hfill\begin{minipage}{0.95\textwidth}
\rule{\textwidth}{1pt}}
{\par\noindent\rule{\textwidth}{1pt}\end{minipage}}
\renewenvironment{proof}{{\bfseries Proof.}}{\qed}
\makeatletter
\renewcommand\@maketitle{
\hfill
\begin{minipage}{0.95\textwidth}
\vskip 2em
\let\footnote\thanks 
{\LARGE \@title \par }
\vskip 1.5em
{\large \@author \par}
\end{minipage}
\vskip 1em \par
}
\makeatother

\newtheorem{theorem}{Theorem}
\newtheorem{theoremP}{Theorem}

\newtheorem{lemma}{Lemma}
\newtheorem{assumption}{Assumption}
\newtheorem{assumptionP}{Assumption}

\title{\textbf{\huge{An efficient methodology to estimate the parameters of a two-dimensional chirp signal model}}}
\author[$\dagger$]{Rhythm Grover}
\author[$\dagger$,$\ddagger$]{Debasis Kundu}
\author[$\dagger$]{Amit Mitra}
\affil[$\dagger$]{Department of Mathematics and Statistics, Indian Institute of Technology Kanpur \\
Kanpur - 208016, India}
\affil[$\ddagger$]{Corresponding author. Email: kundu@iitk.ac.in}
\date{}
\begin{document}
\maketitle

\begin{abstract}
Abstract: In various capacities of statistical signal processing two-dimensional (2-D) chirp models have been considered significantly, particularly in image processing$-$ to model gray-scale and texture images, magnetic resonance imaging, optical imaging etc. In this paper we address the problem of estimation of the unknown parameters of a 2-D chirp model under the assumption that the errors are independently and identically distributed (i.i.d.). The key attribute of the proposed estimation procedure is that it is computationally more efficient than the least squares estimation method. Moreover, the proposed estimators are observed to have the same asymptotic properties as the least squares estimators, thus providing computational effectiveness without any compromise on the efficiency of the estimators. We extend the propounded estimation method to provide a sequential procedure to estimate the unknown parameters of a 2-D chirp model with multiple components 
and under the assumption of i.i.d.\ errors we study the large sample properties of these sequential estimators. Simulation studies and a synthetic data analysis show that the proposed estimators perform satisfactorily.

\end{abstract}
\section{Introduction}\label{sec:Intro}
A two-dimensional (2-D) chirp model has the following mathematical expression:
\begin{equation}\begin{split}\label{multiple_comp_model}
y(m,n)  = \sum_{k=1}^{p} \{A_k^0 \cos(\alpha_k^0 m + \beta_k^0 m^2 + \gamma_k^0 n + \delta_k^0 n^2) + B_k^0 \sin(\alpha_k^0 m + \beta_k^0 m^2 + \gamma_k^0 n + \delta_k^0 n^2)\} + X(m,n);\\
m = 1, \ldots, M; n = 1, \ldots, N.
\end{split}\end{equation}
Here, $y(m,n)$ is the observed signal data, and the parameters $A_k^0$s, $B_k^0$s are the amplitudes, $\alpha_k^0$s, $\gamma_k^0$s are the frequencies and $\beta_k^0$s, $\delta_k^0$s are the frequency rates. The random component $X(m,n)$ accounts for the noise component of the observed signal. In this paper, we assume that $X(m,n)$ is an independently and identically distributed (i.i.d.) random field. \\ \par

It can be seen that the model admits a decomposition of two components$-$ the deterministic component and the random component. The deterministic component represents a gray-scale texture and the random component makes the model more realistic for practical realisation. For illustration, we simulate data with a fixed set of model parameters. Figure \ref{fig:true_signal} represents the gray-scale texture corresponding to the simulated data without the noise component and Figure \ref{fig:noisy_signal} represents the contaminated texture image corresponding to the simulated data with the noise component. This clearly suggests that the 2-D chirp signal models can be used effectively in modelling and analysing black and white texture images. 
\begin{figure}[H]
\centering
\begin{minipage}{0.5\textwidth}
\includegraphics[scale = 0.5]{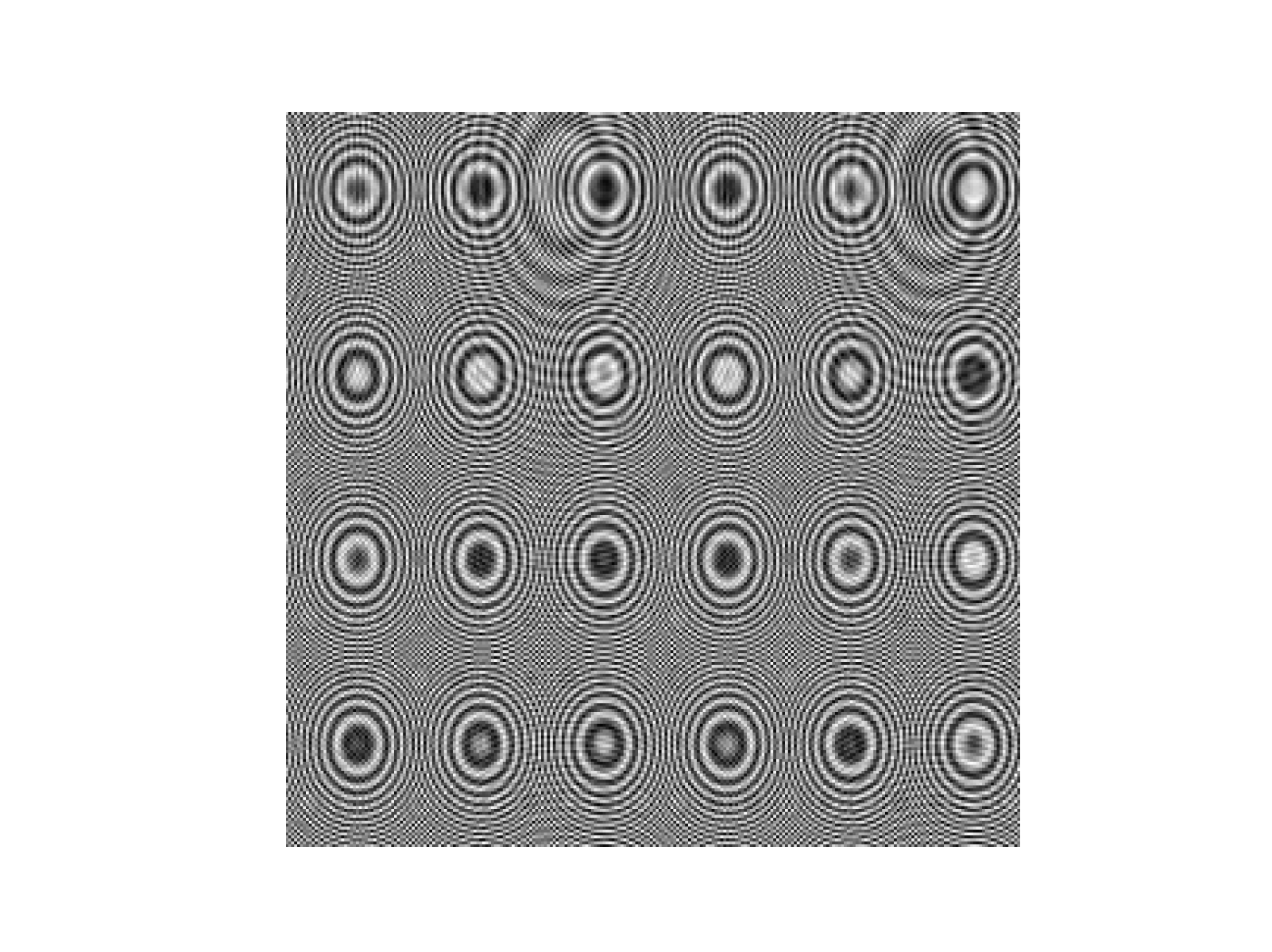}
\caption{Original texture.}
\label{fig:true_signal}
\end{minipage}%
\begin{minipage}{0.5\textwidth}
\includegraphics[scale= 0.5]{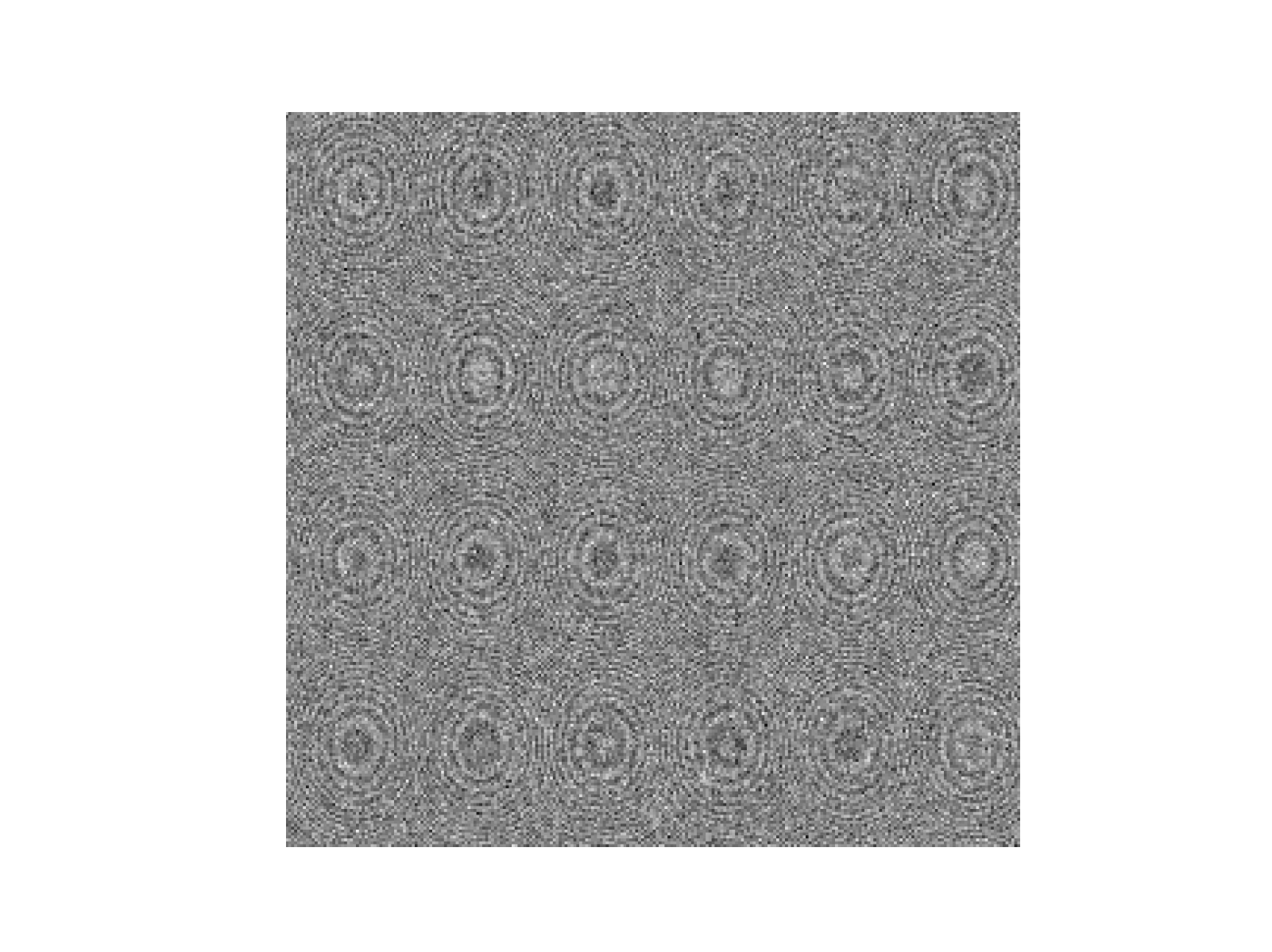}
\caption{Noisy texture.}
\label{fig:noisy_signal}
\end{minipage}
\end{figure}
Apart from the applications in image analysis, these signals are commonly observed in mobile telecommunications, surveillance systems, in radars and sonars etc. For more details on the applications, one may see the works of Francos and Friedlander \cite{1998}, \cite{1999},  Simeunovi$\acute{c}$ and Djurovi$\acute{c}$ \cite{2016} and Zhang et al.\ \cite{2008} and the references cited therein. \\ \par 

Parameter estimation of a 2-D chirp signal is an important statistical signal processing problem. Recently Zhang et al.\ \cite{2008}, Lahiri et al.\ \cite{2013_2} and Grover et al.\ \cite{2018_2} proposed some estimation methods of note. For instance, Zhang et al.\ \cite{2008} proposed an algorithm based on the product cubic phase function for the estimation of the frequency rates of the 2-D chirp signals under low signal to noise ratio and the assumption of stationary errors. They conducted simulations to verify the performance of the proposed estimation algorithm, however there was no study of the theoretical properties of the proposed estimators. Lahiri et al.\ \cite{2015} suggested the least squares estimation method. They observed that the least squares estimators (LSEs) of the unknown parameters of this model  are strongly consistent and asymptotically normally distributed under the assumption of stationary additive errors. The rates of convergence of the amplitude estimates were observed to be $M^{-1/2}N^{-1/2}$, of the frequencies estimates, they are $M^{-3/2}N^{-1/2}$ and $M^{-1/2}N^{-3/2}$ and of the frequency rate estimates, they are $M^{-5/2}N^{-1/2}$ and $M^{-1/2}N^{-5/2}$. Grover et al.\ \cite{2018_2} proposed the approximate least squares estimators (ALSEs), obtained by maximising a periodogram-type function and under the same stationary error assumptions, they observed that ALSEs are strongly consistent and asymptotically equivalent to the LSEs.\\ \par

A chirp signal is a particular case of the polynomial phase signal when the phase is a quadratic polynomial. Although work on parameter estimation of the aforementioned 2-D chirp model is rather limited, several authors have considered the more generalised version of this model$-$the 2-D polynomial phase signal model. For references, see Djurovi\'c et al.\ \cite{2010}, Djurovi\'c \cite{2017_2}, Francos and Friedlander \cite{1998,1999}, Friedlander and Francos \cite{1996}, Lahiri and Kundu \cite{2017}, Simeunovi\'c et al.\ \cite{2014},  Simeunovi$\acute{c}$ and Djurovi$\acute{c}$ \cite{2016} and Djurovi\'c and Simeunovi\'c \cite{2018_3}. \\ \par 

In this paper, we address the problem of parameter estimation of a one-component 2-D chirp model as well as the more general multiple-component 2-D chirp model. We put forward two methods for this purpose. The key characteristic of the proposed estimation method is that it reduces the foregoing 2-D chirp model into two 1-D chirp models. Thus, instead of fitting a 2-D chirp model, we are required to fit two 1-D chirp models to the given data matrix. For the fitting, we use a simple modification of the least squares estimation method. The proposed algorithm is numerically more efficient than the usual least squares estimation method proposed by Lahiri et al. \cite{2015}. For instance, for a one-component 2-D chirp model, to estimate the parameters using these algorithms, we need to solve two 2-D optimisation problems as opposed to a 4-D optimisation problem in the case of finding the LSEs. This also leads to curtailment of the number of grid points required to find the initial values of the non-linear parameters as the 4-D grid search required in case of the computation of the usual LSEs or ALSEs, reduces to two 2-D grid searches. Therefore, instead of searching along a grid mesh consisting of $M^3N^3$ points, we need to search among only $M^3 + N^3$ points, which is much more feasible to execute computationally. In essence, the contributions of this paper are three-fold:
\begin{enumerate}
\item [1.] We put forward a computationally efficient algorithm for the estimation of the unknown parameters of 2-D chirp signal models as a practical alternative to the usual least squares estimation method.
\item [2.] We examine the asymptotic properties of the proposed estimators under the assumption of i.i.d.\ errors and observe that the proposed estimators are strongly consistent and asymptotically normally distributed. In fact, they are observed to be asymptotically equivalent to the corresponding LSEs. When the errors are assumed to be Gaussian, the asymptotic variance-covariance matrix of the proposed estimators coincides with asymptotic Cram\'er-Rao lower bound. 
\item [3.] We conduct simulation experiments and analyse a synthetic texture (see Figure \ref{fig:noisy_signal}) to assess the effectiveness of the proposed estimators. 
\end{enumerate}

The rest of the paper is organised as follows. In the next section, we provide some preliminary results required to study the asymptotic properties of the proposed estimators. In Section \ref{sec:One_component_model}, we consider a one-component 2-D chirp model and state the model assumptions, some notations and present the proposed algorithms along with the asymptotic properties of the proposed estimators. In Section \ref{sec:Multiple_component_model}, we extend the algorithm and develop a sequential procedure to estimate the parameters of a multiple-component 2-D chirp model. We also study the asymptotic properties of the proposed sequential estimators in this section. We perform numerical experiments for different model parameters in Section \ref{sec:Simulation_studies} and analyse a synthetic data for illustration in Section \ref{sec:data_analysis}. Finally, we conclude the paper in Section \ref{sec:Conclusion} and we provide the proofs of all the theoretical claims in the appendices. 

\section{Preliminary Results}\label{sec:preliminary _results}
In this section, we provide the asymptotic results obtained for the usual LSEs of the unknown parameters of a 1-D chirp model by Lahiri et al. \cite{2015}. These results are later exploited to prove the asymptotic normality of the proposed estimators. 
\subsection{One-component 1-D Chirp Model}\label{sec:one_comp_1D}
Consider a 1-D chirp model with the following mathematical expression:
\begin{equation}\label{one_comp_1D_model}
y(t) = A^0 \cos(\alpha^0 t + \beta^0 t^2) + B^0 \sin(\alpha^0 t + \beta^0 t^2) + X(t).
\end{equation}
Here $y(t)$ is the observed data at time points $t = 1, \ldots, n$, $A^0$, $B^0$ are the amplitudes and $\alpha^0$ is the frequency and $\beta^0$ is the frequency rate parameter. $\{X(t)\}_{t=1}^{n}$ is the sequence of error random variables. \\

\noindent The LSEs of $\alpha^0$ and $\beta^0$ can be obtained by minimising the following reduced error sum of squares:
\begin{equation*}
R_{n}(\alpha, \beta) = Q_n(\hat{A}, \hat{B}, \alpha, \beta) = \textit{\textbf{Y}}^{\top}(\textbf{I} - \textbf{P}_{\textbf{Z}_n}(\alpha, \beta))\textit{\textbf{Y}}
\end{equation*}
where, $$Q_n(A, B, \alpha, \beta) = (\textit{\textbf{Y}} - \textbf{Z}_n(\alpha, \beta)\boldsymbol{\phi})^{\top}(\textit{\textbf{Y}} - \textbf{Z}_n(\alpha, \beta)\boldsymbol{\phi}),$$ is the error sum of squares, $\textbf{P}_{\textbf{Z}_n}(\alpha, \beta) = \textbf{Z}_n(\alpha, \beta)(\textbf{Z}_n(\alpha, \beta)^{\top}\textbf{Z}_n(\alpha, \beta))^{-1}\textbf{Z}_n(\alpha, \beta)^{\top}$ is the projection matrix on the column space of the matrix $\textbf{Z}_n(\alpha, \beta)$,
\begin{equation}\label{Z_definition}
\textbf{Z}_n(\alpha, \beta) = \begin{bmatrix}
\cos(\alpha + \beta) & \sin(\alpha + \beta) \\
\vdots & \vdots \\
\cos(n\alpha + n^2\beta) & \sin(n\alpha + n^2\beta) \\
\end{bmatrix},
\end{equation} $\textit{\textbf{Y}} = \begin{bmatrix}
y(1) \ldots y(n)
\end{bmatrix}^{\top}$ is the observed data vector and $\boldsymbol{\phi}  = \begin{bmatrix} A & B \end{bmatrix}^{T}$ is the the vector of linear parameters.  \\ \par
\noindent Following are the assumptions, we make on the error component and the parameters of model \eqref{one_comp_1D_model}:
\begin{assumptionP}\label{assump:P1} X(t) is a sequence of i.i.d.\ random variables with mean zero, variance $\sigma^2$ and finite fourth order moment.
\end{assumptionP}

\begin{assumptionP}\label{assump:P2} $(A^0, B^0, \alpha^0, \beta^0)$ is an interior point of the parameter space $\boldsymbol{\Theta} = (-K, K) \times (-K, K) \times (0, \pi) \times (0, \pi),$ where $K$ is a positive real number and ${A^0}^2 + {B^0}^2 > 0.$
\end{assumptionP}

\begin{theoremP}\label{theorem:preliminary_result_1}
Let us denote $\textit{\textbf{R}}'_{n}(\alpha, \beta)$ as the first derivative vector and $\textit{\textbf{R}}''_{n}(\alpha, \beta)$ as the second derivative matrix of the function $R_{n}(\alpha, \beta)$. Then, under the assumptions \ref{assump:P1} and \ref{assump:P2}, we have:
\begin{equation}\begin{split}\label{prelim_LSE_first_derivative_one_comp}
-\textit{\textbf{R}}'_{n}(\alpha^0, \beta^0)\boldsymbol{\Delta} \rightarrow \boldsymbol{\mathcal{N}}_2(\textbf{0}, 2\sigma^2 \boldsymbol{\Sigma}^{-1}),
\end{split}\end{equation} 
\begin{equation}\begin{split}\label{prelim_LSE_second_derivative_one_comp}
\boldsymbol{\Delta}\textit{\textbf{R}}''_{n}(\alpha^0, \beta^0)\boldsymbol{\Delta} \rightarrow \boldsymbol{\Sigma}^{-1}.
\end{split}\end{equation} 
Here, $\boldsymbol{\Delta} = \textnormal{diag}(\frac{1}{n\sqrt{n}}, \frac{1}{n^2\sqrt{n}})$, 
\begin{equation}\label{Sigma_definition}
\boldsymbol{\Sigma} = \frac{2}{{A^0}^2 + {B^0}^2} \begin{bmatrix}
  96 & -90 \\
 -90 & 90 \\ 
\end{bmatrix} \textnormal{ and }
\end{equation}
 \begin{equation}\label{Sigma_inverse_definition}
  \boldsymbol{\Sigma}^{-1} = \begin{bmatrix}
  \frac{{A^0}^2 + {B^0}^2}{12} & \frac{{A^0}^2 + {B^0}^2}{12} \\
  \frac{{A^0}^2 + {B^0}^2}{12} & \frac{4({A^0}^2 + {B^0}^2)}{45} \\ 
\end{bmatrix}.
\end{equation}
The notation $\boldsymbol{\mathcal{N}}_2(\boldsymbol{\mu}, \boldsymbol{\mathcal{V}})$ means bivariate normally distributed with mean vector $\boldsymbol{\mu}_{2 \times 1}$ and variance-covariance matrix $\boldsymbol{\mathcal{V}}_{2 \times 2}$. 
\end{theoremP}
\begin{proof}
This proof follows from Theorem 2 of Lahiri et al. \cite{2015}. \\
\end{proof}

\subsection{Multiple-component 1-D Chirp Model}\label{sec:multiple_comp_1D}
Now we consider a 1-D chirp model with multiple components, mathematically expressed as follows:
\begin{equation*}
y(t) = \sum_{k=1}^{p} \{A_k^0 \cos(\alpha_k^0 t + \beta_k^0 t^2) + B_k^0 \sin(\alpha_k^0 t + \beta_k^0 t^2)\} + X(t); \ t = 1, \ldots, n.
\end{equation*}
Here, $A_k^0$s, $B_k^0$s are the amplitudes, $\alpha_k^0$s are the frequencies and $\beta_k^0$ are the frequency rates, the parameters that characterise the observed signal $y(t)$ and $X(t)$ is the random noise component. \\

Lahiri et al.\ \cite{2015} suggested a sequential procedure to estimate the unknown parameters of the above model. We discuss in brief, the proposed sequential procedure and then state some of the asymptotic results they established, germane to our work.
\begin{description}
\item \namedlabel{itm:step1} {\textbf{Step 1:}}\ The first step of the sequential method is to estimate the non-linear parameters of the first component of the model, $\alpha_1^0$ and $\beta_1^0$, say $\hat{\alpha}_1$ and $\hat{\beta}_1$ by minimising the following reduced error sum of squares:
\begin{equation*}
R_{1,n}(\alpha, \beta) = \textit{\textbf{Y}}^{\top}(\textbf{I} - \textbf{P}_{\textbf{Z}_n}(\alpha, \beta))\textit{\textbf{Y}}
\end{equation*}
with respect to $\alpha$ and $\beta$ simultaneously.
\item \namedlabel{itm:step2} {\textbf{Step 2:}}\  Then the first component linear parameter estimates, $\hat{A}_1$ and $\hat{B}_1$ are obtained using the separable linear regression of Richards \cite{1961} as follows:
\begin{equation*}
\begin{bmatrix}
\hat{A}_1 \\ \hat{B}_1
\end{bmatrix} = [\textbf{Z}_n(\hat{\alpha}_1, \hat{\beta}_1)^{\top}\textbf{Z}_n(\hat{\alpha}_1, \hat{\beta}_1)]^{-1}\textbf{Z}_n(\hat{\alpha}_1, \hat{\beta}_1)^{\top}\textit{\textbf{Y}}.
\end{equation*}
\item \namedlabel{itm:step3} {\textbf{Step 3:}}\  Once we have the estimates of the first component parameters, we take out its effect from the original signal and obtain a new data vector as follows:
\begin{equation*}
\textit{\textbf{Y}}_1 = \textit{\textbf{Y}} - \textbf{Z}_n(\hat{\alpha}_1, \hat{\beta}_1)\begin{bmatrix}
\hat{A}_1 \\ \hat{B}_1
\end{bmatrix}.
\end{equation*}
\item \namedlabel{itm:step4} {\textbf{Step 4:}}\  Then the estimates of the second component parameters are obtained by using the new data vector and following the same procedure and the process is repeated $p$ times.
\end{description} 
\noindent Under the Assumption \ref{assump:P1} on the error random variables and the following assumption on the parameters:
\begin{assumptionP}\label{assump:P3} $(A_k^0, B_k^0, \alpha_k^0, \beta_k^0)$ is an interior point of $\boldsymbol{\Theta}$, for all $k = 1, \ldots, p$ and the frequencies and the frequency rates are such that $(\alpha_i^0, \beta_i^0) \neq (\alpha_j^0, \beta_j^0)$ $\forall i \neq j$. 
\end{assumptionP}
\begin{assumptionP}\label{assump:P4} $A_k^0$s and $B_k^0$s satisfy the following relationship:
\begin{equation*}
K^2 > {A_{1}^{0}}^2 + {B_{1}^{0}}^2 > {A_{2}^{0}}^2 + {B_{2}^{0}}^2 > \ldots > {A_{p}^{0}}^2 + {B_{p}^{0}}^2 > 0,
\end{equation*}
\end{assumptionP}
\noindent we have the following results.
\begin{theoremP}\label{theorem:preliminary_result_3}
Let us denote $\textit{\textbf{R}}'_{k,n}(\alpha, \beta)$ as the first derivative vector and $\textit{\textbf{R}}''_{k,n}(\alpha, \beta)$ as the second derivative matrix of the function $R_{k,n}(\alpha, \beta)$, $k = 1, \ldots, p$. Then, under the assumptions \ref{assump:P1}, \ref{assump:P3} and \ref{assump:P4}:
\begin{equation}\begin{split}\label{prelim_LSE_first_derivative_multiple_comp_1}
-\frac{1}{\sqrt{n}}\textit{\textbf{R}}'_{k,n}(\alpha^0, \beta^0)\boldsymbol{\Delta} \rightarrow 0,
\end{split}\end{equation}
\begin{equation}\begin{split}\label{prelim_LSE_first_derivative_multiple_comp_2}
-\textit{\textbf{R}}'_{k,n}(\alpha^0, \beta^0)\boldsymbol{\Delta} \rightarrow \boldsymbol{\mathcal{N}}_2(\textbf{0}, 2\sigma^2 \boldsymbol{\Sigma}_k^{-1}),
\end{split}\end{equation} 
\begin{equation}\begin{split}\label{prelim_LSE_second_derivative_multiple_comp}
\boldsymbol{\Delta}\textit{\textbf{R}}''_{k,n}(\alpha^0, \beta^0)\boldsymbol{\Delta} \rightarrow \boldsymbol{\Sigma}_k^{-1}.
\end{split}\end{equation} 
Here, $\boldsymbol{\Delta}$ is as defined in Theorem \ref{theorem:preliminary_result_1},
\begin{equation}\label{Sigma_k_definition}
\boldsymbol{\Sigma}_k = \frac{2}{{A_k^0}^2 + {B_k^0}^2} \begin{bmatrix}
  96 & -90 \\
 -90 & 90
\end{bmatrix} \textnormal{ and }
\end{equation} 
\begin{equation}\label{Sigma_k_inverse_definition}
 \boldsymbol{\Sigma}_k^{-1} = \begin{bmatrix}
  \frac{{A_k^0}^2 + {B_k^0}^2}{12} & \frac{{A_k^0}^2 + {B_k^0}^2}{12} \\
  \frac{{A_k^0}^2 + {B_k^0}^2}{12} & \frac{4({A_k^0}^2 + {B_k^0}^2)}{45}
\end{bmatrix}.
\end{equation} 
\end{theoremP}
\begin{proof}
The proof of \eqref{prelim_LSE_first_derivative_multiple_comp_1} follows along the same lines as proof of Lemma 4 of Lahiri et al. \cite{2015} and that of  \eqref{prelim_LSE_first_derivative_multiple_comp_2} and \eqref{prelim_LSE_second_derivative_multiple_comp} follows from Theorem 2 of Lahiri et al. \cite{2015}. Note that Lahiri et al. \cite{2015} showed that the sequential LSEs have the same asymptotic distribution as the usual LSEs based on a famous number theory conjecture (see the reference). \\
\end{proof}

\section{One-Component 2-D Chirp Model}\label{sec:One_component_model}
In this section, we provide the methodology to obtain the proposed estimators for the parameters of a one-component 2-D chirp model, mathematically expressed as follows:
\begin{equation}\begin{split}\label{one_comp_model}
y(m,n)  =A^0 \cos(\alpha^0 m + \beta^0 m^2 + \gamma^0 n + \delta^0 n^2) + B^0 \sin(\alpha^0 m + \beta^0 m^2 + \gamma^0 n + \delta^0 n^2) + X(m,n);\\
m = 1, \ldots, M; n = 1, \ldots, N.
\end{split}\end{equation}
Here $y(m,n)$ is the observed data vector and the parameters $A^0$, $B^0$ are the amplitudes, $\alpha^0$, $\gamma^0$ are the frequencies and $\beta^0$, $\delta^0$ are the frequency rates of the signal model. As mentioned in the introduction, $X(m,n)$ accounts for the noise present in the signal.\\

\noindent We will use the following notations: $\boldsymbol{\theta} = (A, B, \alpha, \beta, \gamma, \delta)$ is the parameter vector, \\
$\boldsymbol{\theta}^0 = (A^0, B^0, \alpha^0, \beta^0, \gamma^0, \delta^0)$ is the true parameter vector and $\boldsymbol{\Theta}$ = $(-K,K) \times (-K, K) \times (0,\pi) \times (0,\pi)\times (0,\pi) \times (0,\pi)$ is the parameter space.

\subsection{Proposed Methodology}\label{sec:Estimation_one_comp_LSEs_method}
Let us consider the above-stated 2-D chirp signal model with one-component. Suppose we fix $n = n_0$, then \eqref{one_comp_model} can be rewritten as follows:
\begin{equation}\begin{split}\label{model_n=n0_one_comp}
y(m,n_0) & = A^0 \cos(\alpha^0 m + \beta^0 m^2 + \gamma^0 n_0 + \delta^0 n_0^2) + B^0 \sin(\alpha^0 m + \beta^0 m^2 + \gamma^0 n_0 + \delta^0 n_0^2) + X(m,n_0)\\
& = A^0(n_0) \cos(\alpha^0 m + \beta^0 m^2) + B^0(n_0) \sin(\alpha^0 m + \beta^0 m^2) + X(m,n_0);\ \  m = 1, \cdots, M,
\end{split}\end{equation}
which represents a 1-D chirp model with $A^0(n_0)$, $B^0(n_0)$ as the amplitudes, $\alpha^0$ as the frequency parameter and $\beta^0$ as the frequency rate parameter. Here, 
\begin{equation*}\begin{split}
& A^0(n_0) = \ \ A^0 \cos(\gamma^0 n_0 + \delta^0 n_0^2) + B^0 \sin(\gamma^0 n_0 + \delta^0 n_0^2), \textmd{ and} \\
& B^0(n_0) = -A^0 \sin(\gamma^0 n_0 + \delta^0 n_0^2) + B^0 \cos(\gamma^0 n_0 + \delta^0 n_0^2).
\end{split}\end{equation*}
Thus for each fixed $n_0$ $\in$ $\{1, \ldots, N\}$, we have a 1-D chirp model with the same frequency and frequency rate parameters, though different amplitudes. This 1-D model corresponds to a column of the 2-D data matrix.\\

Our aim is to estimate the non-linear parameters $\alpha^0$ and $\beta^0$ from the columns of the data matrix and one of the most reasonable estimators for this purpose are the least squares estimators. Therefore, the estimators of $\alpha^0$ and $\beta^0$ can be obtained by minimising the following function:
\begin{equation*}\begin{split}
R_M(\alpha, \beta, n_0) & 
 = {\textit{\textbf{Y}}^{\top}_{n_0}}(\textbf{I} - \textbf{P}_{\textbf{Z}_M}(\alpha, \beta))\textit{\textbf{Y}}_{n_0}
\end{split}\end{equation*}
for each $n_0$. Here, $\textit{\textbf{Y}}_{n_0} = \begin{bmatrix}
y[1,n_0] & \ldots & y[M,n_0]
\end{bmatrix}^{\top}$ is the $n_0$\textit{th} column of of the original data matrix, $\textbf{P}_{\textbf{Z}_M}(\alpha, \beta) = \textbf{Z}_M(\alpha, \beta)(\textbf{Z}_M(\alpha, \beta)^{\top}\textbf{Z}_M(\alpha, \beta))^{-1}\textbf{Z}_M(\alpha, \beta)^{\top}$ is the projection matrix on the column space of the matrix $\textbf{Z}_M(\alpha, \beta)$ and the matrix $\textbf{Z}_M(\alpha, \beta)$ can be obtained by replacing $n$ by $M$ in \eqref{Z_definition}. This process involves minimising $N$ 2-D functions corresponding to the N columns of the matrix. Thus, for computational efficiency, we propose to minimise the following function instead:
\begin{equation}\begin{split}\label{reduced_ess_alpha_beta}
R^{(1)}_{MN}(\alpha, \beta)  & = \sum_{n_0 = 1}^{N}R_M(\alpha, \beta, n_0) = \sum_{n_0 = 1}^{N} \textit{\textbf{Y}}^{\top}_{n_0}(\textbf{I} -  \textbf{P}_{\textbf{Z}_M}(\alpha, \beta))\textit{\textbf{Y}}_{n_0}
\end{split}\end{equation}
with respect to $\alpha$ and $\beta$ simultaneously and obtain $\hat{\alpha}$ and $\hat{\beta}$ which reduces the estimation process to solving only one 2-D optimisation problem. Note that since the errors are assumed to be i.i.d.\ replacing these $N$ functions by their sum is justifiable.

Similarly, we can obtain the estimates, $\hat{\gamma}$ and $\hat{\delta}$, of $\gamma^0$ and $\delta^0$, by minimising the following criterion function:
 \begin{equation}\begin{split}\label{reduced_ess_gamma_delta}
R^{(2)}_{MN}(\gamma, \delta)  & = \sum_{m_0 = 1}^{M}R_N(\gamma, \delta, m_0) = \sum_{m_0 = 1}^{M} \textit{\textbf{Y}}^{\top}_{m_0}(\textbf{I} -  \textbf{P}_{\textbf{Z}_N}(\gamma, \delta))\textit{\textbf{Y}}_{m_0}
\end{split}\end{equation}
with respect to $\gamma$ and $\delta$ simultaneously. The data vector $\textit{\textbf{Y}}_{m_0} = \begin{bmatrix}
y[m_0,1] & \ldots & y[m_0,N] 
\end{bmatrix}^{\top}$, is the $m_0$\textit{th} row of the data matrix, $m_0 = 1, \ldots, M$, $\textbf{P}_{\textbf{Z}_N}(\gamma, \delta)$ is the projection matrix on the column space of the matrix $\textbf{Z}_N(\gamma, \delta)$ and the matrix $\textbf{Z}_{N}(\gamma, \delta)$ can be obtained by replacing $n$ by $N$ and $\alpha$ and $\beta$ by $\gamma$ and $\delta$ respectively in the matrix $\textbf{Z}_n(\alpha, \beta)$, defined in \eqref{Z_definition}.\\ 

Once we have the estimates of the non-linear parameters, we estimate the linear parameters by the usual least squares regression technique as proposed by Lahiri et al. \cite{2015}:
\begin{equation*}
\begin{bmatrix}
\hat{A} \\ \hat{B} \end{bmatrix} = [\textit{\textbf{W}}(\hat{\alpha}, \hat{\beta}, \hat{\gamma}, \hat{\delta})^{T} \textit{\textbf{W}}(\hat{\alpha}, \hat{\beta}, \hat{\gamma}, \hat{\delta})]^{-1}\textit{\textbf{W}}(\hat{\alpha}, \hat{\beta}, \hat{\gamma}, \hat{\delta})^{T}\textit{\textbf{Y}}.
\end{equation*}
Here, $\textit{\textbf{Y}}_{M N \times 1} =  \left[\begin{array}{ccccccc}y(1, 1) & \ldots & y(M, 1) & \ldots &  y(1, N) & \ldots & y(M, N)\end{array}\right]^{T}$ is the observed data vector, and
\begin{equation}\label{W_matrix}
\textit{\textbf{W}}(\alpha, \beta, \gamma, \delta)_{M N \times 2} = \left[\begin{array}{cc}\cos(\alpha + \beta + \gamma + \delta) & \sin(\alpha + \beta + \gamma + \delta) \\ \cos(2\alpha + 4\beta + \gamma + \delta) & \sin(2\alpha + 4\beta + \gamma + \delta) \\ \vdots & \vdots \\ \cos(M \alpha + M^2 \beta + \gamma + \delta) &  \sin(M \alpha + M^2 \beta + \gamma + \delta) \\ \vdots & \vdots \\  \cos(\alpha + \beta + N \gamma + N^2 \delta) & \sin(\alpha + \beta + N \gamma + N^2 \delta) \\ \cos(2\alpha + 4\beta + N \gamma + N^2 \delta) & \sin(2\alpha + 4\beta + N \gamma + N^2 \delta)\\\vdots & \vdots \\ \cos(M \alpha + M^2 \beta + N \gamma + N^2 \delta) & \sin(M \alpha + M^2 \beta + N \gamma + N^2 \delta)\end{array}\right].
\end{equation}

\noindent We make the following assumptions on the error component and the model parameters before we examine the asymptotic properties of the proposed estimators:
\begin{assumption}\label{assump:1} X(m,n) is a double array sequence of i.i.d.\ random variables with mean zero, variance $\sigma^2$ and finite fourth order moment.
\end{assumption}

\begin{assumption}\label{assump:2} The true parameter vector $\boldsymbol{\theta}^0$ is an interior point of the parametric space $\boldsymbol{\Theta}_1$, and ${A^0}^2 + {B^0}^2 > 0$.
\end{assumption}

\subsection{Consistency}\label{sec:Estimation_one_comp_LSEs_consistency}
The results obtained on the consistency of the proposed estimators are presented in the following theorems:
\begin{theorem}\label{theorem:consistency_alpha_beta_one_comp_LSE}
Under assumptions \ref{assump:1} and \ref{assump:2}, $\hat{\alpha}$ and  $\hat{\beta}$ are strongly consistent estimators of $\alpha^0$ and $\beta^0$ respectively, that is,
\begin{equation*}\begin{split}
\hat{\alpha} \xrightarrow{a.s.} \alpha^0 \textmd{ as } M \rightarrow \infty.\\
\hat{\beta} \xrightarrow{a.s.} \beta^0 \textmd{ as } M \rightarrow \infty.
\end{split}\end{equation*} 
\end{theorem}
\begin{proof}
See \nameref{appendix:A}.\\
\end{proof}
\begin{theorem}\label{theorem:consistency_gamma_delta_one_comp_LSE}
Under assumptions \ref{assump:1} and \ref{assump:2}, $\hat{\gamma}$ and  $\hat{\delta}$ are strongly consistent estimators of $\gamma^0$ and $\delta^0$ respectively, that is,
\begin{equation*}\begin{split}
\hat{\gamma} \xrightarrow{a.s.} \gamma^0 \textmd{ as } N \rightarrow \infty.\\
\hat{\delta} \xrightarrow{a.s.} \delta^0 \textmd{ as } N \rightarrow \infty.
\end{split}\end{equation*} 
\end{theorem}
\begin{proof}
This proof follows along the same lines as the proof of Theorem \ref{theorem:consistency_alpha_beta_one_comp_LSE}.\\
\end{proof}
\subsection{Asymptotic distribution.}\label{sec:Estimation_one_comp_LSEs_distribution}
The following theorems provide the asymptotic distributions of the proposed estimators:
\begin{theorem}\label{theorem:asymp_dist_one_comp_LSE_alpha_beta}
If the assumptions, \ref{assump:1} and \ref{assump:2} are satisfied, then 
\begin{equation*}\begin{split}
\begin{bmatrix}
(\hat{\alpha} - \alpha^0) & , & (\hat{\beta} - \beta^0)
\end{bmatrix}\textit{\textbf{D}}_1^{-1} \xrightarrow{d} \boldsymbol{\mathcal{N}}_2(\textbf{0}, 2\sigma^2 \boldsymbol{\Sigma}) \textmd{ as } M \rightarrow \infty.\\
\end{split}\end{equation*}
Here, $ \mathbf{D}_1 = \textnormal{diag}(M^{\frac{-3}{2}}N^{\frac{-1}{2}}, M^{\frac{-5}{2}}N^{\frac{-1}{2}})$ and $\boldsymbol{\Sigma}$ is as defined in \eqref{Sigma_definition}.
\end{theorem}
 \begin{proof}
See \nameref{appendix:A}.\\
\end{proof}

\begin{theorem}\label{theorem:asymp_dist_one_comp_LSE_gamma_delta}
If the assumptions, \ref{assump:1} and \ref{assump:2} are satisfied, then 
\begin{equation*}\begin{split}
\begin{bmatrix}
(\hat{\gamma} - \gamma^0) & , & (\hat{\delta} - \delta^0)
\end{bmatrix}\textit{\textbf{D}}_2^{-1} \xrightarrow{d} \boldsymbol{\mathcal{N}}_2(\textbf{0}, 2\sigma^2 \boldsymbol{\Sigma}) \textmd{ as } N \rightarrow \infty.\\
\end{split}\end{equation*}
Here, $ \mathbf{D}_2 = \textnormal{diag}( M^{\frac{-1}{2}}N^{\frac{-3}{2}}, M^{\frac{-1}{2}}N^{\frac{-5}{2}})$ and $\boldsymbol{\Sigma}$ is as defined in \eqref{Sigma_definition}.
\end{theorem}
\begin{proof}
This proof follows along the same lines as the proof of Theorem \ref{theorem:asymp_dist_one_comp_LSE_alpha_beta}.\\
\end{proof}

\noindent The asymptotic distributions of $(\hat{\alpha}, \hat{\beta})$ and $(\hat{\gamma}, \hat{\delta})$ are observed to be the same as those of the corresponding LSEs. Thus, we get the same efficiency as that of the LSEs without going through the exhaustive process of actually computing the LSEs.  

\section{Multiple-Component 2-D Chirp model}\label{sec:Multiple_component_model}
In this section, we consider the multipl-component 2-D chirp model with $p$ number of components, with the mathematical expression of the model as given in \eqref{multiple_comp_model}. Although estimation of $p$ is an important problem, in this paper we deal with the estimation of the other important parameters characterising the observed signal, the amplitudes, the frequencies and the frequency rates, assuming $p$ to be known. We propose a sequential procedure to estimate these parameters. The main idea supporting the proposed sequential procedure is same as that behind the ones proposed by Prasad et al. \cite{2008_2}  for a sinusoidal model and Lahiri et al. \cite{2015} and Grover et al. \cite{2018} for a chirp model$-$ the orthogonality of different regressor vectors. Along with the computationally efficiency, the sequential method provides estimators with the same rates of convergence as the LSEs.\\ 

\subsection{Proposed Sequential Algorithm}\label{sec:Estimation_multiple_comp_LSEs_method}
\noindent The following algorithm is a simple extension of the method proposed to obtain the estimators for a one-component 2-D model in Section \ref{sec:Estimation_one_comp_LSEs_method}:

\begin{description}
\item \namedlabel{itm:step1} {\textbf{Step 1:}}\ Compute $\hat{\alpha}_1$ and $\hat{\beta}_1$ by minimising the following function:
\begin{equation*}
R^{(1)}_{1, MN}(\alpha, \beta)  = \sum_{n_0 = 1}^{N} \textit{\textbf{Y}}^{\top}_{n_0}(\textbf{I} - \textbf{P}_{\textbf{Z}_M}(\alpha, \beta))\textit{\textbf{Y}}_{n_0}
\end{equation*} 
with respect to $\alpha$ and $\beta$ simultaneously.  
\item \namedlabel{itm:step2} {\textbf{Step 2:}}\ Compute $\hat{\gamma}_1$ and $\hat{\delta}_1$ by minimising the function:
\begin{equation*}
R^{(2)}_{1,MN}(\gamma, \delta) = \sum_{m_0 = 1}^{M} \textit{\textbf{Y}}^{\top}_{m_0}(\textbf{I} - \textbf{P}_{\textbf{Z}_N}(\alpha, \beta))\textit{\textbf{Y}}_{m_0}
\end{equation*} 
with respect to $\gamma$ and $\delta$ simultaneously. 
\item \namedlabel{itm:step3} {\textbf{Step 3:}}\ Once the nonlinear parameters of the first component of the model are estimated, estimate the linear parameters $A_1^0$ and $B_1^0$ by the usual least squares estimation technique:
\begin{equation*}
\begin{bmatrix}
\hat{A}_1 \\ \hat{B}_1 \end{bmatrix} = [\textit{\textbf{W}}(\hat{\alpha}_1, \hat{\beta}_1, \hat{\gamma}_1, \hat{\delta}_1)^{T} \textit{\textbf{W}}(\hat{\alpha}_1, \hat{\beta}_1, \hat{\gamma}_1, \hat{\delta}_1)]^{-1}\textit{\textbf{W}}(\hat{\alpha}_1, \hat{\beta}_1, \hat{\gamma}_1, \hat{\delta}_1)^{T}\textit{\textbf{Y}}.
\end{equation*}
Here, $\textit{\textbf{Y}}_{M N \times 1} =  \left[\begin{array}{ccccccc}y(1, 1) & \ldots & y(M, 1) & \ldots &  y(1, N) & \ldots & y(M, N)\end{array}\right]^{T}$ is the observed data vector, and the matrix $\textit{\textbf{W}}(\hat{\alpha}_1, \hat{\beta}_1, \hat{\gamma}_1, \hat{\delta}_1)$ can be obtained by replacing $\alpha$, $\beta$, $\gamma$ and $\delta$ by $\hat{\alpha}_1$, $\hat{\beta}_1$, $\hat{\gamma}_1$ and $\hat{\delta}_1$ respectively in \eqref{W_matrix}.
\item \namedlabel{itm:step4} {\textbf{Step 4:}}\ Eliminate the effect of the first component from the original data and construct new data as follows:
\begin{equation}\begin{split}\label{second_stage_data}
y_1(m,n) = y(m,n) - \hat{A}_1 \cos(\hat{\alpha}_1 m + \hat{\beta}_1 m^2 + \hat{\gamma}_1 n + \hat{\delta}_1 n^2) - \hat{B}_1 \sin(\hat{\alpha}_1 m + \hat{\beta}_1 m^2 + \hat{\gamma}_1 n + \hat{\delta}_1 n^2); \\
m = 1, \ldots, M;\ n = 1, \ldots, N.
\end{split}\end{equation} 
\item \namedlabel{itm:step5} {\textbf{Step 5:}}\ Using the new data, estimate the parameters of the second component by following the same procedure.
\item \namedlabel{itm:step6} {\textbf{Step 6:}}\  Continue this process until all the parameters are estimated. \\
\end{description}

\noindent In the following subsections, we examine the asymptotic properties of the proposed estimators under the assumptions \ref{assump:1}, \ref{assump:P4} and the following assumption on the parameters:

\begin{assumption}\label{assump:3} $\boldsymbol{\theta}_{k}^{0}$ is an interior point of $\boldsymbol{\Theta}_1$, for all $k = 1, \ldots, p$ and the frequencies $\alpha_{k}^0s$, $\gamma_{k}^0s$ and the frequency rates $\beta_{k}^0s$, $\delta_{k}^0s$ are such that $(\alpha_i^0, \beta_i^0) \neq (\alpha_j^0, \beta_j^0)$ and $(\gamma_i^0, \delta_i^0)$ $\neq$ $(\gamma_j^0, \delta_j^0)$ $\forall i \neq j$. 
\end{assumption}

\subsection{Consistency.}\label{sec:Estimation_multiple_comp_LSEs_consistency}
Through the following theorems, we proclaim the consistency of the proposed estimators when the number of components, $p$ is unknown.

\begin{theorem}\label{theorem:consistency_alphak_betak_multiple_comp_LSE}
If assumptions \ref{assump:1}, \ref{assump:3} and \ref{assump:P4} are satisfied, then the following results hold true for $ 1 \leqslant k \leqslant p$:
\begin{equation*}\begin{split}
& \hat{\alpha}_k \xrightarrow{a.s.} \alpha_k^0 \textmd{ as } M \rightarrow \infty, \\
& \hat{\beta}_k \xrightarrow{a.s.} \beta_k^0 \textmd{ as } M \rightarrow \infty.
\end{split}\end{equation*}
\end{theorem}
\begin{proof}
See \nameref{appendix:B}.\\
\end{proof}
\begin{theorem}\label{theorem:consistency_gammak_deltak_multiple_comp_LSE}
If assumptions \ref{assump:1}, \ref{assump:3} and \ref{assump:P4} are satisfied, then the following results hold true for $ 1 \leqslant k \leqslant p$:
\begin{equation*}\begin{split}
& \hat{\gamma}_k \xrightarrow{a.s.} \gamma_k^0 \textmd{ as } N \rightarrow \infty, \\
& \hat{\delta}_k \xrightarrow{a.s.} \delta_k^0 \textmd{ as } N \rightarrow \infty.
\end{split}\end{equation*}
\end{theorem}
\begin{proof}
This proof can be obtained along the same lines as proof of \textnormal{Theorem \ref{theorem:consistency_alphak_betak_multiple_comp_LSE}}. \\
\end{proof}

\begin{theorem}\label{theorem:limit_A_p+1_B_p+1}
If the assumptions \ref{assump:1}, \ref{assump:3} and \ref{assump:P4} are satisfied, and if $\hat{A_k}$,  $\hat{B_k}$, $\hat{\alpha_k}$, $\hat{\beta_k}$, $\hat{\gamma_k}$ and $\hat{\delta_k}$ are the estimators obtained at the $k$-th step, then \\
for $k \leqslant p$, 
\begin{equation*}\begin{split}
& \hat{A_k} \xrightarrow{ a.s } A_k^0 \textmd{ as } \textnormal{min}\{M, N\} \rightarrow \infty \\ 
& \hat{B_k} \xrightarrow{ a.s } B_k^0 \textmd{ as } \textnormal{min}\{M, N\} \rightarrow \infty,
\end{split}\end{equation*}
and for $k > p$,
\begin{equation*}\begin{split}
& \hat{A_k} \xrightarrow{ a.s } 0 \textmd{ as } \textnormal{min}\{M, N\} \rightarrow \infty \\ 
& \hat{B_k} \xrightarrow{ a.s } 0 \textmd{ as } \textnormal{min}\{M, N\} \rightarrow \infty.
\end{split}\end{equation*}
\end{theorem}
\begin{proof}
This proof follows from the proof of Theorem 2.4.4 of Lahiri \cite{2015}.\\
\end{proof}

\noindent Note that we do not know the number of components in practice. The problem of estimation of $p$ is an important problem though we have not considered it here. From the above theorem, it is clear that if the number of components of the fitted model is less than or same as the true number of components, $p$, then the amplitude estimators converge to their true values almost surely, else if it is more than $p$, then the amplitude estimators upto the $p$-th step converge to the true values and past that, they converge to zero almost surely. Thus, this result can be used a criterion to estimate the number $p$. However, this might not work in low signal to noise ratio scenarios.  
\subsection{Asymptotic distribution.}\label{sec:Estimation_multiple_comp_LSEs_distribution}
\begin{theorem}\label{theorem:asymp_dist_multiple_comp_LSE_alphak_betak}
If assumptions \ref{assump:1}, \ref{assump:3} and \ref{assump:P4} are satisfied, then for $1 \leqslant k \leqslant p:$
\begin{equation*}\begin{split}
& \begin{bmatrix}
(\hat{\alpha}_k - \alpha_k^0) & , & (\hat{\beta}_k - \beta_k^0) 
\end{bmatrix}\textit{\textbf{D}}_1^{-1} \xrightarrow{d} \boldsymbol{\mathcal{N}}_2(\textbf{0}, 2\sigma^2 \boldsymbol{\Sigma}_{k}) \textmd{ as } M \rightarrow \infty.\\
\end{split}\end{equation*}
Here $\textit{\textbf{D}}_1$ is as defined in Theorem \ref{theorem:asymp_dist_one_comp_LSE_alpha_beta} and $\boldsymbol{\Sigma}_{k}$ is as defined in \eqref{Sigma_k_definition}.
\end{theorem}
\begin{proof}
See \nameref{appendix:B}. \\
\end{proof}

\begin{theorem}\label{theorem:asymp_dist_multiple_comp_LSE_gammak_deltak}
If the assumptions, \ref{assump:1}, \ref{assump:3} and \ref{assump:P4} are satisfied, then 
\begin{equation*}\begin{split}
\begin{bmatrix}
(\hat{\gamma}_k - \gamma_k^0) & , & (\hat{\delta}_k - \delta_k^0)
\end{bmatrix}\textit{\textbf{D}}_2^{-1} \xrightarrow{d} \boldsymbol{\mathcal{N}}_2(\textbf{0}, 2\sigma^2 \boldsymbol{\Sigma}_k) \textmd{ as } N \rightarrow \infty.\\
\end{split}\end{equation*}
Here $\textit{\textbf{D}}_2$ is as defined in Theorem \ref{theorem:asymp_dist_one_comp_LSE_gamma_delta} and $\boldsymbol{\Sigma}_{k}$ is as defined in \eqref{Sigma_k_definition}.
\end{theorem}
 \begin{proof}
This proof follows along the same lines as the proof of Theorem \ref{theorem:asymp_dist_multiple_comp_LSE_alphak_betak}.\\
\end{proof}

\section{Numerical Experiments and Simulated Data Analysis}\label{sec:Simulations}
\subsection{Numerical Experiments}\label{sec:Simulation_studies}
We perform simulations to examine the performance of the proposed estimators. We consider the following two cases:
\begin{description}
\item \namedlabel{itm:case1} {\textbf{Case I:}}\ When the data are generated from a one-component model \eqref{one_comp_model}, with the following set of parameters: \\
$A^0 = 2$, $B^0 = 3$, $\alpha^0 = 1.5$, $\beta^0 = 0.5$, $\gamma^0 = 2.5$ and $\delta^0 = 0.75$. \label{case_1_simulations} 
\item \namedlabel{itm:case2} {\textbf{Case II:}}\ When the data are generated from a two components model \eqref{multiple_comp_model}, with the following set of parameters:\\
 $A_1^0 = 5$, $B_1^0 = 4$, $\alpha_1^0 = 2.1$, $\beta_1^0 = 0.1$, $\gamma_1^0 = 1.25$ and $\delta_1^0 = 0.25$, $A_2^0 = 3$, $B_2^0 = 2$, $\alpha_2^0 = 1.5$, $\beta_2^0 = 0.5$, $\gamma_2^0 = 1.75$ and $\delta_2^0 = 0.75$. \label{case_2_simulations} 
\end{description}
The noise used in the simulations is generated from Gaussian distribution with mean 0 and variance $\sigma^2$. Also, different values of the error variance, $\sigma^2$ and sample sizes, $M$ and $N$ are considered. We estimate the parameters using the proposed estimation technique as well as the least squares estimation  technique for \textbf{Case I} and for \textbf{Case II}, the proposed sequential technique and the sequential least squares technique proposed by Lahiri \cite{2015} are employed for comparison. For each case, the procedure is replicated 1000 times and the average values of the estimates, the average biases and the mean square errors (MSEs) are reported. The collation of the MSEs and the theoretical asymptotic variances (Avar) exhibits the efficacy of the proposed estimation method.  
\subsubsection{One-component simulation results}\label{sec:one_comp_simulations}
In Table \ref{table:1}-Table \ref{table:4}, the results obtained through simulations for \textbf{Case I} are presented. It is observed that as $M$ and $N$ increase, the average estimates get closer to the true values, the average biases decrease and the MSEs decrease as well, thus verifying consistency of the proposed estimates. Also, the biases and the MSEs of both types of estimates increase as the error variance increases. The MSEs of the proposed estimators are of the same order as those of the LSEs and thus are well-matched with the corresponding asymptotic variances.
\begin{table}[]
\centering
\resizebox{0.75\textwidth}{!}{\begin{tabular}{cc|cccc|cccc}
\hline
\multicolumn{2}{c|}{Parameters} & $\alpha$ & $\beta$ & $\gamma$ & $\delta$ & $\alpha$ & $\beta$ & $\gamma$ & $\delta$ \\
\multicolumn{2}{c|}{True values} & 1.5 & 0.5 & 2.5 & 0.75 & 1.5 & 0.5 & 2.5 & 0.75 \\ \hline
$\sigma$ &        & \multicolumn{4}{c|}{Proposed estimators}                               & \multicolumn{4}{c}{Usual LSEs}    \\ \hline
0.10  &     Avg   & 1.5000     &  0.5000     &   2.5000     &   0.7500    &   1.5000     &   0.5000    &   2.5000     &   0.7500   \\
      &     Bias  & 3.26e-05   &  -1.23e-06  &   2.37e-05   &  -1.26e-06  &   2.86e-05   &  -1.07e-06  &   2.13e-05   &  -1.17e-06 \\
      &     MSE   & 9.01e-07   &  1.21e-09   &   8.34e-07   &   1.14e-09  &   8.75e-07   &   1.19e-09  &   8.02e-07   &   1.10e-09 \\
      &     Avar  & 7.56e-07   &  1.13e-09   &   7.56e-07   &   1.13e-09  &   7.56e-07   &   1.13e-09  &   7.56e-07   &   1.13e-09 \\  \hline
0.50  &     Avg   & 1.4997     &  0.5000     &   2.4999     &   0.7500    &   1.4998     &   0.5000    &   2.5000     &   0.7500   \\
      &     Bias  & -2.78e-04  &   1.05e-05  &   -6.85e-05  &   3.20e-06  &   -2.01e-04  &   7.90e-06  &   -7.23e-06  &   8.39e-07 \\
      &     MSE   & 2.37e-05   &  3.18e-08   &   2.17e-05   &   3.11e-08  &   2.17e-05   &   2.97e-08  &   2.07e-05   &   2.95e-08 \\
      &     Avar  & 1.89e-05   &  2.84e-08   &   1.89e-05   &   2.84e-08  &   1.89e-05   &   2.84e-08  &   1.89e-05   &   2.84e-08 \\  \hline
1.00  &     Avg   & 1.5004     &  0.5000     &   2.4998     &   0.7500    &   1.5004     &   0.5000    &   2.4998     &   0.7500   \\
      &     Bias  & 4.11e-04   &  -1.77e-05  &   -2.24e-04  &   8.00e-06  &   3.63e-04   &  -1.60e-05  &   -2.16e-04  &   7.57e-06 \\
      &     MSE   & 9.54e-05   &  1.22e-07   &   8.92e-05   &   1.25e-07  &   8.92e-05   &   1.17e-07  &   8.48e-05   &   1.18e-07 \\
      &     Avar  & 7.56e-05   &  1.13e-07   &   7.56e-05   &   1.13e-07  &   7.56e-05   &   1.13e-07  &   7.56e-05   &   1.13e-07 \\
\hline
\end{tabular}}
\caption{Estimates of the parameters of model \eqref{one_comp_model} when M =  N = 25}
\label{table:1}
\end{table}

\begin{table}[]
\centering
\resizebox{0.75\textwidth}{!}{\begin{tabular}{cc|cccc|cccc}
\hline
\multicolumn{2}{c|}{Parameters} & $\alpha$ & $\beta$ & $\gamma$ & $\delta$ & $\alpha$ & $\beta$ & $\gamma$ & $\delta$ \\
\multicolumn{2}{c|}{True values} & 1.5 & 0.5 & 2.5 & 0.75 & 1.5 & 0.5 & 2.5 & 0.75 \\ \hline
$\sigma$ &        & \multicolumn{4}{c|}{Proposed estimators}                               & \multicolumn{4}{c}{Usual LSEs}    \\ \hline
0.10  &     Avg   & 1.5000      &   0.5000       &   2.5000      &    0.7500     &    1.5000      &    0.5000     &    2.5000      &    0.7500   \\
      &     Bias  & 5.53e-06    &   -1.25e-07    &   2.33e-06    &   -4.65e-08   &    2.51e-06    &   -7.96e-08   &    -1.89e-06   &    2.57e-08 \\
      &     MSE   & 4.88e-08    &   1.83e-11     &   5.09e-08    &    1.88e-11   &    4.14e-08    &    1.54e-11   &    4.65e-08    &    1.72e-11 \\
      &     Avar  & 4.73e-08    &   1.77e-11     &   4.73e-08    &    1.77e-11   &    4.73e-08    &    1.77e-11   &    4.73e-08    &    1.77e-11 \\  \hline
0.50  &     Avg   & 1.5000      &   0.5000       &   2.5000      &    0.7500     &    1.5000      &    0.5000     &    2.5000      &    0.7500   \\
      &     Bias  & -3.57e-05   &    4.91e-07    &   -4.47e-05   &    7.83e-07   &    2.56e-06    &   -2.61e-07   &    -4.19e-05   &    6.84e-07 \\
      &     MSE   & 1.35e-06    &   4.93e-10     &   1.31e-06    &    4.78e-10   &    1.16e-06    &    4.18e-10   &    1.18e-06    &    4.34e-10 \\
      &     Avar  & 1.18e-06    &   4.43e-10     &   1.18e-06    &    4.43e-10   &    1.18e-06    &    4.43e-10   &    1.18e-06    &    4.43e-10 \\  \hline
1.00  &     Avg   & 1.5000      &   0.5000       &   2.5000      &    0.7500     &    1.5000      &    0.5000     &    2.5000      &    0.7500   \\
      &     Bias  & 2.11e-05    &   -2.41e-07    &   -2.42e-05   &    2.35e-07   &    5.55e-06    &   -1.92e-09   &    -2.37e-05   &    2.45e-07 \\
      &     MSE   & 5.36e-06    &   1.92e-09     &   5.03e-06    &    1.77e-09   &    4.38e-06    &    1.56e-09   &    4.53e-06    &    1.60e-09 \\
      &     Avar  & 4.73e-06    &   1.77e-09     &   4.73e-06    &    1.77e-09   &    4.73e-06    &    1.77e-09   &    4.73e-06    &    1.77e-09 \\
\hline
\end{tabular}}
\caption{Estimates of the parameters of model \eqref{one_comp_model} when M =  N = 50}
\label{table:2}
\end{table}

\begin{table}[]
\centering
\resizebox{0.75\textwidth}{!}{\begin{tabular}{cc|cccc|cccc}
\hline
\multicolumn{2}{c|}{Parameters} & $\alpha$ & $\beta$ & $\gamma$ & $\delta$ & $\alpha$ & $\beta$ & $\gamma$ & $\delta$ \\
\multicolumn{2}{c|}{True values} & 1.5 & 0.5 & 2.5 & 0.75 & 1.5 & 0.5 & 2.5 & 0.75 \\ \hline
$\sigma$ &        & \multicolumn{4}{c|}{Proposed estimators}                               & \multicolumn{4}{c}{Usual LSEs}    \\ \hline
0.10  &     Avg   & 1.5000      &   0.5000       &   2.5000      &    0.7500     &    1.5000      &    0.5000     &    2.5000      &    0.7500   \\
      &     Bias  & -5.61e-06   &    5.15e-08    &   9.53e-07    &   7.95e-10    &   -5.73e-06    &   5.26e-08    &   -2.45e-07    &   1.47e-08  \\
      &     MSE   & 9.77e-09    &   1.60e-12     &   1.02e-08    &    1.65e-12   &    9.48e-09    &    1.56e-12   &    9.10e-09    &    1.48e-12 \\
      &     Avar  & 9.34e-09    &   1.56e-12     &   9.34e-09    &    1.56e-12   &    9.34e-09    &    1.56e-12   &    9.34e-09    &    1.56e-12 \\  \hline
0.50  &     Avg   & 1.5000      &   0.5000       &   2.5000      &    0.7500     &    1.5000      &    0.5000     &    2.5000      &    0.7500   \\
      &     Bias  & -3.55e-05   &    3.80e-07    &   1.73e-06    &   -1.05e-07   &    -3.45e-05   &    3.63e-07   &    4.93e-06    &   -1.46e-07 \\
      &     MSE   & 2.45e-07    &   4.00e-11     &   2.39e-07    &    3.86e-11   &    2.01e-07    &    3.29e-11   &    1.79e-07    &    2.96e-11 \\
      &     Avar  & 2.33e-07    &   3.89e-11     &   2.33e-07    &    3.89e-11   &    2.33e-07    &    3.89e-11   &    2.33e-07    &    3.89e-11 \\  \hline
1.00  &     Avg   & 1.5000      &   0.5000       &   2.5000      &    0.7500     &    1.5000      &    0.5000     &    2.5000      &    0.7500   \\
      &     Bias  & -4.93e-06   &    7.23e-08    &   4.93e-05    &   -6.60e-07   &    -1.67e-05   &    2.28e-07   &    2.78e-05    &   -3.89e-07 \\
      &     MSE   & 1.01e-06    &   1.67e-10     &   1.06e-06    &    1.74e-10   &    8.29e-07    &    1.39e-10   &    7.77e-07    &    1.31e-10 \\
      &     Avar  & 9.34e-07    &   1.56e-10     &   9.34e-07    &    1.56e-10   &    9.34e-07    &    1.56e-10   &    9.34e-07    &    1.56e-10 \\
\hline
\end{tabular}}
\caption{Estimates of the parameters of model \eqref{one_comp_model} when M =  N = 75}
\label{table:3}
\end{table}
\begin{table}[]
\centering
\resizebox{0.75\textwidth}{!}{\begin{tabular}{cc|cccc|cccc}
\hline
\multicolumn{2}{c|}{Parameters} & $\alpha$ & $\beta$ & $\gamma$ & $\delta$ & $\alpha$ & $\beta$ & $\gamma$ & $\delta$ \\
\multicolumn{2}{c|}{True values} & 1.5 & 0.5 & 2.5 & 0.75 & 1.5 & 0.5 & 2.5 & 0.75 \\ \hline
$\sigma$ &        & \multicolumn{4}{c|}{Proposed estimators}                               & \multicolumn{4}{c}{Usual LSEs}    \\ \hline
0.10  &     Avg   & 1.5000      &   0.5000       &   2.5000      &    0.7500     &    1.5000      &    0.5000     &    2.5000      &    0.7500   \\
      &     Bias  & -5.76e-07   &    -5.65e-10   &    6.02e-07   &    2.59e-09   &    -6.66e-07   &    -1.33e-10   &   4.29e-07   &    3.89e-09  \\
      &     MSE   & 3.23e-09    &   2.92e-13     &   3.00e-09    &    2.85e-13   &    2.47e-09    &    2.28e-13   &    2.87e-09    &    2.74e-13 \\
      &     Avar  & 2.95e-09    &   2.77e-13     &   2.95e-09    &    2.77e-13   &    2.95e-09    &    2.77e-13   &    2.95e-09    &    2.77e-13 \\  \hline
0.50  &     Avg   & 1.5000      &   0.5000       &   2.5000      &    0.7500     &    1.5000      &    0.5000     &    2.5000      &    0.7500   \\
      &     Bias  & -5.41e-06   &    5.31e-08    &   1.12e-05    &   -1.10e-07   &    -1.07e-06   &    1.56e-08   &    1.38e-05    &   -1.34e-07 \\
      &     MSE   & 8.11e-08    &   7.28e-12     &   7.52e-08    &    6.83e-12   &    5.41e-08    &    5.03e-12   &    5.54e-08    &    5.18e-12 \\
      &     Avar  & 7.38e-08    &   6.92e-12     &   7.38e-08    &    6.92e-12   &    7.38e-08    &    6.92e-12   &    7.38e-08    &    6.92e-12 \\  \hline
1.00  &     Avg   & 1.5000      &   0.5000       &   2.5000      &    0.7500     &    1.5000      &    0.5000     &    2.5000      &    0.7500   \\
      &     Bias  & -1.98e-05   &    1.63e-07    &   1.43e-05    &   -8.73e-08   &    -8.54e-06   &    6.29e-08   &    1.12e-05    &   -5.96e-08 \\
      &     MSE   & 2.83e-07    &   2.56e-11     &   2.96e-07    &    2.75e-11   &    1.91e-07    &    1.77e-11   &    2.07e-07    &    1.97e-11 \\
      &     Avar  & 2.95e-07    &   2.77e-11     &   2.95e-07    &    2.77e-11   &    2.95e-07    &    2.77e-11   &    2.95e-07    &    2.77e-11 \\
\hline
\end{tabular}}
\caption{Estimates of the parameters of model \eqref{one_comp_model} when M =  N = 100}
\label{table:4}
\end{table}


\subsubsection{Two component simulation results}\label{sec:two_comp_simulations}
We present the simulation results for \textbf{Case II} in Table \ref{table:5}-Table \ref{table:8}. From these tables, it is evident that the average estimates are quite close to the true values. The results also verify consistency of the proposed sequential estimators. It is also observed that the MSEs of the parameter estimates of the first component are mostly of the same order as the corresponding theoretical variances while those of the second component have exactly the same order as the corresponding asymptotic variances. 
\begin{table}[]
\centering
\resizebox{\textwidth}{!}{\begin{tabular}{c|c|c|cccc|cccc}
\hline
$\sigma$ &                  &            & \multicolumn{4}{c|}{Proposed sequential estimates}     & \multicolumn{4}{c}{Sequential LSEs}               \\ \hline
0.10    &    First Component  & Parameters & $\alpha_1$ & $\beta_1$ & $\gamma_1$ & $\delta_1$ & $\alpha_1$ & $\beta_1$ & $\gamma_1$ & $\delta_1$ \\
        &                     & True values & 2.1        & 0.1       & 1.25       & 0.25       & 2.1        & 0.1       & 1.25       & 0.25       \\ \hline
        &                     &  Average & 2.1016      &     0.0998     &     1.2614     &      0.2500     &     2.1031     &      0.0998     &     1.2565     &     0.2500    \\
        &                     &  Bias    & 1.63e-03    &    -1.81e-04   &     1.14e-02   &     -4.71e-05   &     3.05e-03   &     -1.76e-04   &     6.46e-03   &     3.94e-05  \\
        &                     &  MSE     & 2.92e-06    &     3.30e-08   &     1.31e-04   &      2.60e-09   &     9.59e-06   &      3.12e-08   &     4.20e-05   &     1.91e-09  \\
        &                     &  AVar    & 2.40e-07    &     3.60e-10   &     2.40e-07   &      3.60e-10   &     2.40e-07   &      3.60e-10   &     2.40e-07   &     3.60e-10  \\ \hline
        &   Second Component  & Parameters  & $\alpha_2$ & $\beta_2$ & $\gamma_2$ & $\delta_2$ & $\alpha_2$ & $\beta_2$ & $\gamma_2$ & $\delta_2$ \\ \hline
        &                     & True values & 1.5        & 0.5       & 1.75       & 0.75       & 1.5        & 0.5       & 1.75       & 0.75       \\ \hline
        &                     &  Average & 1.5018      &     0.5000     &     1.7520     &      0.7499     &     1.5017     &      0.5000     &     1.7510     &     0.7500    \\
        &                     &  Bias    & 1.83e-03    &    -1.92e-05   &     1.98e-03   &     -6.41e-05   &     1.68e-03   &     -2.19e-05   &     1.03e-03   &     -2.95e-05 \\
        &                     &  MSE     & 4.19e-06    &     1.54e-09   &     4.96e-06   &      5.53e-09   &     3.61e-06   &      1.59e-09   &     1.97e-06   &     2.12e-09  \\
        &                     &  AVar    & 7.56e-07    &     1.13e-09   &     7.56e-07   &      1.13e-09   &     7.56e-07   &      1.13e-09   &     7.56e-07   &     1.13e-09  \\ \hline
0.50    &    First Component  & Parameters  & $\alpha_1$ & $\beta_1$ & $\gamma_1$ & $\delta_1$ & $\alpha_1$ & $\beta_1$ & $\gamma_1$ & $\delta_1$ \\
        &                     & True values & 2.1        & 0.1       & 1.25       & 0.25       & 2.1        & 0.1       & 1.25       & 0.25       \\ \hline
        &                     &  Average & 2.1017      &     0.0998     &     1.2613     &      0.2500     &     2.1031     &      0.0998     &     1.2563     &     0.2500    \\
        &                     &  Bias    & 1.71e-03    &    -1.83e-04   &     1.13e-02   &     -4.38e-05   &     3.13e-03   &     -1.78e-04   &     6.32e-03   &     4.34e-05  \\
        &                     &  MSE     & 8.92e-06    &     4.14e-08   &     1.35e-04   &      1.18e-08   &     1.60e-05   &      3.98e-08   &     4.61e-05   &     1.07e-08  \\
        &                     &  AVar    & 5.99e-06    &     8.99e-09   &     5.99e-06   &      8.99e-09   &     5.99e-06   &      8.99e-09   &     5.99e-06   &     8.99e-09  \\ \hline
        &   Second Component  & Parameters & $\alpha_2$ & $\beta_2$ & $\gamma_2$ & $\delta_2$ & $\alpha_2$ & $\beta_2$ & $\gamma_2$ & $\delta_2$ \\
        &                     & True values & 1.5        & 0.5       & 1.75       & 0.75       & 1.5        & 0.5       & 1.75       & 0.75       \\ \hline
        &                     &  Average & 1.5017      &     0.5000     &     1.7522     &      0.7499     &     1.5015     &      0.5000     &     1.7512     &     0.7500    \\
        &                     &  Bias    & 1.69e-03    &    -1.05e-05   &     2.16e-03   &     -6.94e-05   &     1.51e-03   &     -1.25e-05   &     1.17e-03   &     -3.35e-05 \\
        &                     &  MSE     & 2.54e-05    &     3.37e-08   &     3.15e-05   &      4.22e-08   &     2.35e-05   &      3.18e-08   &     2.40e-05   &     3.31e-08  \\
        &                     &  AVar    & 1.89e-05    &     2.84e-08   &     1.89e-05   &      2.84e-08   &     1.89e-05   &      2.84e-08   &     1.89e-05   &     2.84e-08  \\ \hline
1.00    &    First Component  & Parameters  & $\alpha_1$ & $\beta_1$ & $\gamma_1$ & $\delta_1$ & $\alpha_1$ & $\beta_1$ & $\gamma_1$ & $\delta_1$ \\
        &                     & True values & 2.1        & 0.1       & 1.25       & 0.25       & 2.1        & 0.1       & 1.25       & 0.25       \\ \hline
        &                     &  Average & 2.1015      &     0.0998     &     1.2616     &      0.2499     &     2.1029     &      0.0998     &     1.2567     &     0.2500    \\
        &                     &  Bias    & 1.54e-03    &    -1.79e-04   &     1.16e-02   &     -5.63e-05   &     2.93e-03   &     -1.72e-04   &     6.67e-03   &     3.07e-05  \\
        &                     &  MSE     & 2.98e-05    &     6.77e-08   &     1.65e-04   &      4.44e-08   &     3.72e-05   &      6.70e-08   &     6.97e-05   &     3.66e-08  \\
        &                     &  AVar    & 2.40e-05    &     3.60e-08   &     2.40e-05   &      3.60e-08   &     2.40e-05   &      3.60e-08   &     2.40e-05   &     3.60e-08   \\ \hline
        &   Second Component  & Parameters  & $\alpha_2$ & $\beta_2$ & $\gamma_2$ & $\delta_2$ & $\alpha_2$ & $\beta_2$ & $\gamma_2$ & $\delta_2$ \\
        &                     & True values & 1.5        & 0.5       & 1.75       & 0.75       & 1.5        & 0.5       & 1.75       & 0.75       \\ \hline
        &                     &  Average & 1.5018      &     0.5000     &     1.7516     &      0.7499     &     1.5018     &      0.5000     &     1.7507     &     0.7500     \\
        &                     &  Bias    & 1.79e-03    &    -1.57e-05   &     1.64e-03   &     -5.47e-05   &     1.75e-03   &     -2.25e-05   &     7.18e-04   &     -2.07e-05  \\
        &                     &  MSE     & 9.84e-05    &     1.41e-07   &     1.20e-04   &      1.66e-07   &     9.40e-05   &      1.35e-07   &     1.01e-04   &     1.40e-07   \\
        &                     &  AVar    & 7.56e-05    &     1.13e-07   &     7.56e-05   &      1.13e-07   &     7.56e-05   &      1.13e-07   &     7.56e-05   &     1.13e-07   \\ \hline
\end{tabular}}
\caption{Estimates of the parameters of model \eqref{multiple_comp_model} when M =  N = 25}
\label{table:5}
\end{table}

\begin{table}[]
\centering
\resizebox{\textwidth}{!}{\begin{tabular}{c|c|c|cccc|cccc}
\hline
$\sigma$ &                  &            & \multicolumn{4}{c|}{Proposed sequential estimates}     & \multicolumn{4}{c}{Sequential LSEs}               \\  \hline
0.10    &    First Component  & Parameters & $\alpha_1$ & $\beta_1$ & $\gamma_1$ & $\delta_1$ & $\alpha_1$ & $\beta_1$ & $\gamma_1$ & $\delta_1$ \\
        &                     & True values & 2.1        & 0.1       & 1.25       & 0.25       & 2.1        & 0.1       & 1.25       & 0.25       \\  \hline
        &                     &  Average & 2.1006      &     0.1000     &     1.2567     &      0.2499     &     2.1011     &      0.1000     &     1.2572     &     0.2499    \\
        &                     &  Bias    & 6.07e-04    &    -8.61e-06   &     6.70e-03   &     -1.11e-04   &     1.07e-03   &     -1.11e-05   &     7.20e-03   &     -1.22e-04 \\
        &                     &  MSE     & 3.85e-07    &     8.02e-11   &     4.49e-05   &      1.23e-08   &     1.15e-06   &      1.29e-10   &     5.18e-05   &     1.49e-08  \\
        &                     &  AVar    & 1.50e-08    &     5.62e-12   &     1.50e-08   &      5.62e-12   &     1.50e-08   &      5.62e-12   &     1.50e-08   &     5.62e-12  \\  \hline
        &   Second Component  & Parameters  & $\alpha_2$ & $\beta_2$ & $\gamma_2$ & $\delta_2$ & $\alpha_2$ & $\beta_2$ & $\gamma_2$ & $\delta_2$ \\
        &                     & True values & 1.5        & 0.5       & 1.75       & 0.75       & 1.5        & 0.5       & 1.75       & 0.75       \\  \hline
        &                     &  Average & 1.5006      &     0.5000     &     1.7506     &      0.7500     &     1.5008     &      0.5000     &     1.7507     &     0.7500    \\
        &                     &  Bias    & 6.09e-04    &    -1.05e-05   &     5.56e-04   &     -9.98e-06   &     7.51e-04   &     -1.36e-05   &     6.83e-04   &     -1.20e-05 \\
        &                     &  MSE     & 4.23e-07    &     1.30e-10   &     3.60e-07   &      1.18e-10   &     6.13e-07   &      2.03e-10   &     5.17e-07   &     1.61e-10  \\
        &                     &  AVar    & 4.73e-08    &     1.77e-11   &     4.73e-08   &      1.77e-11   &     4.73e-08   &      1.77e-11   &     4.73e-08   &     1.77e-11  \\  \hline
0.50    &    First Component  & Parameters  & $\alpha_1$ & $\beta_1$ & $\gamma_1$ & $\delta_1$ & $\alpha_1$ & $\beta_1$ & $\gamma_1$ & $\delta_1$ \\
        &                     & True values & 2.1        & 0.1       & 1.25       & 0.25       & 2.1        & 0.1       & 1.25       & 0.25       \\  \hline
        &                     &  Average & 2.1006      &     0.1000     &     1.2567     &      0.2499     &     2.1011     &      0.1000     &     1.2572     &     0.2499   \\
        &                     &  Bias    & 6.27e-04    &    -9.05e-06   &     6.72e-03   &     -1.11e-04   &     1.09e-03   &     -1.15e-05   &     7.22e-03   &     -1.22e-04\\
        &                     &  MSE     & 8.09e-07    &     2.36e-10   &     4.56e-05   &      1.25e-08   &     1.58e-06   &      2.78e-10   &     5.25e-05   &     1.51e-08 \\
        &                     &  AVar    & 3.75e-07    &     1.40e-10   &     3.75e-07   &      1.40e-10   &     3.75e-07   &      1.40e-10   &     3.75e-07   &     1.40e-10 \\   \hline
        &   Second Component  & Parameters & $\alpha_2$ & $\beta_2$ & $\gamma_2$ & $\delta_2$ & $\alpha_2$ & $\beta_2$ & $\gamma_2$ & $\delta_2$ \\
        &                     & True values & 1.5        & 0.5       & 1.75       & 0.75       & 1.5        & 0.5       & 1.75       & 0.75       \\ \hline
        &                     &  Average & 1.5006      &     0.5000     &     1.7505     &      0.7500     &     1.5008     &      0.5000     &     1.7506     &     0.7500   \\
        &                     &  Bias    & 6.09e-04    &    -1.07e-05   &     5.04e-04   &     -8.84e-06   &     7.52e-04   &     -1.38e-05   &     6.26e-04   &     -1.07e-05\\
        &                     &  MSE     & 1.68e-06    &     5.94e-10   &     1.50e-06   &      5.49e-10   &     1.79e-06   &      6.29e-10   &     1.62e-06   &     5.72e-10 \\
        &                     &  AVar    & 1.18e-06    &     4.43e-10   &     1.18e-06   &      4.43e-10   &     1.18e-06   &      4.43e-10   &     1.18e-06   &     4.43e-10 \\   \hline
1.00    &    First Component  & Parameters  & $\alpha_1$ & $\beta_1$ & $\gamma_1$ & $\delta_1$ & $\alpha_1$ & $\beta_1$ & $\gamma_1$ & $\delta_1$ \\
        &                     & True values & 2.1        & 0.1       & 1.25       & 0.25       & 2.1        & 0.1       & 1.25       & 0.25       \\  \hline
        &                     &  Average & 2.1006      &     0.1000     &     1.2567     &      0.2499     &     2.1011     &      0.1000     &     1.2572     &     0.2499   \\
        &                     &  Bias    & 5.98e-04    &    -8.36e-06   &     6.70e-03   &     -1.11e-04   &     1.06e-03   &     -1.09e-05   &     7.20e-03   &     -1.22e-04\\
        &                     &  MSE     & 1.92e-06    &     6.44e-10   &     4.66e-05   &      1.29e-08   &     2.64e-06   &      6.72e-10   &     5.34e-05   &     1.55e-08 \\
        &                     &  AVar    & 1.50e-06    &     5.62e-10   &     1.50e-06   &      5.62e-10   &     1.50e-06   &      5.62e-10   &     1.50e-06   &     5.62e-10  \\  \hline
        &   Second Component  & Parameters  & $\alpha_2$ & $\beta_2$ & $\gamma_2$ & $\delta_2$ & $\alpha_2$ & $\beta_2$ & $\gamma_2$ & $\delta_2$ \\
        &                     & True values & 1.5        & 0.5       & 1.75       & 0.75       & 1.5        & 0.5       & 1.75       & 0.75       \\  \hline
        &                     &  Average & 1.5006      &     0.5000     &     1.7507     &      0.7500     &     1.5008     &      0.5000     &     1.7508     &     0.7500    \\
        &                     &  Bias    & 6.49e-04    &    -1.10e-05   &     6.55e-04   &     -1.15e-05   &     7.70e-04   &     -1.36e-05   &     7.92e-04   &     -1.38e-05 \\
        &                     &  MSE     & 5.75e-06    &     2.09e-09   &     4.50e-06   &      1.65e-09   &     5.66e-06   &      2.03e-09   &     4.68e-06   &     1.68e-09  \\
        &                     &  AVar    & 4.73e-06    &     1.77e-09   &     4.73e-06   &      1.77e-09   &     4.73e-06   &      1.77e-09   &     4.73e-06   &     1.77e-09  \\  \hline
\end{tabular}}
\caption{Estimates of the parameters of model \eqref{multiple_comp_model} when M =  N = 50}
\label{table:6}
\end{table}

\begin{table}[]
\centering
\resizebox{\textwidth}{!}{\begin{tabular}{c|c|c|cccc|cccc}
\hline
$\sigma$ &                  &            & \multicolumn{4}{c|}{Proposed sequential estimates}     & \multicolumn{4}{c}{Sequential LSEs}               \\  \hline
0.10    &    First Component  & Parameters & $\alpha_1$ & $\beta_1$ & $\gamma_1$ & $\delta_1$ & $\alpha_1$ & $\beta_1$ & $\gamma_1$ & $\delta_1$ \\
        &                     & True values & 2.1        & 0.1       & 1.25       & 0.25       & 2.1        & 0.1       & 1.25       & 0.25       \\ \hline
        &                     &  Average & 2.1000      &    0.1000      &   1.2528     &     0.2500      &    2.1000      &    0.1000     &    1.2528      &   0.2500   \\
        &                     &  Bias    & -6.14e-06   &    -8.37e-07   &    2.81e-03  &     -3.31e-05   &    -2.85e-05   &    9.79e-08   &    2.80e-03    &   -3.27e-05\\
        &                     &  MSE     & 3.06e-09    &    1.19e-12    &   7.90e-06   &     1.10e-09    &    5.96e-09    &    8.46e-13   &    7.84e-06    &   1.07e-09 \\
        &                     &  AVar    & 2.96e-09    &    4.93e-13    &   2.96e-09   &     4.93e-13    &    2.96e-09    &    4.93e-13   &    2.96e-09    &   4.93e-13 \\ \hline
        &   Second Component  & Parameters  & $\alpha_2$ & $\beta_2$ & $\gamma_2$ & $\delta_2$ & $\alpha_2$ & $\beta_2$ & $\gamma_2$ & $\delta_2$ \\
        &                     & True values & 1.5        & 0.5       & 1.75       & 0.75       & 1.5        & 0.5       & 1.75       & 0.75       \\ \hline
        &                     &  Average & 1.5001      &    0.5000      &   1.7500     &     0.7500      &    1.5001      &    0.5000     &    1.7500      &   0.7500   \\
        &                     &  Bias    & 5.82e-05    &   -5.89e-07    &   5.55e-06   &    1.89e-07     &    5.64e-05    &   -5.67e-07   &    2.25e-05    &   -2.69e-08\\
        &                     &  MSE     & 1.26e-08    &    1.85e-12    &   9.86e-09   &     1.64e-12    &    1.23e-08    &    1.82e-12   &    1.02e-08    &   1.57e-12 \\
        &                     &  AVar    & 9.34e-09    &    1.56e-12    &   9.34e-09   &     1.56e-12    &    9.34e-09    &    1.56e-12   &    9.34e-09    &   1.56e-12 \\ \hline
0.50    &    First Component  & Parameters  & $\alpha_1$ & $\beta_1$ & $\gamma_1$ & $\delta_1$ & $\alpha_1$ & $\beta_1$ & $\gamma_1$ & $\delta_1$ \\
        &                     & True values & 2.1        & 0.1       & 1.25       & 0.25       & 2.1        & 0.1       & 1.25       & 0.25       \\ \hline
        &                     &  Average & 2.1000      &    0.1000      &   1.2528     &     0.2500      &    2.1000      &    0.1000     &    1.2528      &   0.2500    \\
        &                     &  Bias    & 1.36e-05    &   -1.06e-06    &   2.82e-03   &    -3.33e-05    &    -1.21e-05   &    -9.17e-08  &     2.82e-03   &    -3.29e-05\\
        &                     &  MSE     & 7.02e-08    &    1.24e-11    &   8.05e-06   &     1.12e-09    &    7.03e-08    &    1.14e-11   &    8.01e-06    &   1.09e-09  \\
        &                     &  AVar    & 7.40e-08    &    1.23e-11    &   7.40e-08   &     1.23e-11    &    7.40e-08    &    1.23e-11   &    7.40e-08    &   1.23e-11  \\\hline
        &   Second Component  & Parameters & $\alpha_2$ & $\beta_2$ & $\gamma_2$ & $\delta_2$ & $\alpha_2$ & $\beta_2$ & $\gamma_2$ & $\delta_2$ \\
        &                     & True values & 1.5        & 0.5       & 1.75       & 0.75       & 1.5        & 0.5       & 1.75       & 0.75       \\ \hline
        &                     &  Average & 1.5000      &    0.5000      &   1.7500     &     0.7500      &    1.5000      &    0.5000     &    1.7500      &   0.7500   \\
        &                     &  Bias    & 4.54e-05    &   -3.85e-07    &   7.93e-06   &    1.36e-07     &    3.58e-05    &   -2.65e-07   &    2.58e-05    &   -9.44e-08\\
        &                     &  MSE     & 2.50e-07    &    4.08e-11    &   2.49e-07   &     4.03e-11    &    2.21e-07    &    3.65e-11   &    2.42e-07    &   3.91e-11 \\
        &                     &  AVar    & 2.33e-07    &    3.89e-11    &   2.33e-07   &     3.89e-11    &    2.33e-07    &    3.89e-11   &    2.33e-07    &   3.89e-11 \\ \hline
1.00    &    First Component  & Parameters  & $\alpha_1$ & $\beta_1$ & $\gamma_1$ & $\delta_1$ & $\alpha_1$ & $\beta_1$ & $\gamma_1$ & $\delta_1$ \\
        &                     & True values & 2.1        & 0.1       & 1.25       & 0.25       & 2.1        & 0.1       & 1.25       & 0.25       \\ \hline
        &                     &  Average & 2.1000      &    0.1000      &   1.2528     &     0.2500      &    2.1000      &    0.1000     &    1.2528      &   0.2500    \\
        &                     &  Bias    & 1.22e-05    &   -1.10e-06    &   2.82e-03   &    -3.32e-05    &    -9.57e-06   &    -1.62e-07  &     2.81e-03   &    -3.27e-05\\
        &                     &  MSE     & 3.09e-07    &    5.09e-11    &   8.28e-06   &     1.15e-09    &    3.07e-07    &    4.96e-11   &    8.18e-06    &   1.12e-09  \\
        &                     &  AVar    & 2.96e-07    &    4.93e-11    &   2.96e-07   &     4.93e-11    &    2.96e-07    &    4.93e-11   &    2.96e-07    &   4.93e-11   \\\hline
        &   Second Component  & Parameters  & $\alpha_2$ & $\beta_2$ & $\gamma_2$ & $\delta_2$ & $\alpha_2$ & $\beta_2$ & $\gamma_2$ & $\delta_2$ \\
        &                     & True values & 1.5        & 0.5       & 1.75       & 0.75       & 1.5        & 0.5       & 1.75       & 0.75       \\ \hline
        &                     &  Average & 1.5000      &    0.5000      &   1.7501     &     0.7500      &    1.5000      &    0.5000     &    1.7501      &   0.7500   \\
        &                     &  Bias    & 3.16e-05    &   -2.66e-07    &   5.78e-05   &    -6.26e-07    &    2.85e-06    &   7.51e-08    &    6.49e-05    &   -7.12e-07\\
        &                     &  MSE     & 9.43e-07    &    1.52e-10    &   9.58e-07   &     1.55e-10    &    8.36e-07    &    1.36e-10   &    9.41e-07    &   1.52e-10 \\
        &                     &  AVar    & 9.34e-07    &    1.56e-10    &   9.34e-07   &     1.56e-10    &    9.34e-07    &    1.56e-10   &    9.34e-07    &   1.56e-10 \\ \hline
\end{tabular}}
\caption{Estimates of the parameters of model \eqref{multiple_comp_model} when M =  N = 75}
\label{table:7}
\end{table}

\begin{table}[]
\centering
\resizebox{\textwidth}{!}{\begin{tabular}{c|c|c|cccc|cccc}
\hline
$\sigma$ &                  &            & \multicolumn{4}{c|}{Proposed sequential estimates}     & \multicolumn{4}{c}{Sequential LSEs}               \\ \hline
0.10    &    First Component  & Parameters & $\alpha_1$ & $\beta_1$ & $\gamma_1$ & $\delta_1$ & $\alpha_1$ & $\beta_1$ & $\gamma_1$ & $\delta_1$ \\
        &                     & True values & 2.1        & 0.1       & 1.25       & 0.25       & 2.1        & 0.1       & 1.25       & 0.25       \\ \hline
        &                     &  Average & 2.0992     &   0.1000    &   1.2507     &   0.2500    &   2.0995     &   0.1000    &   1.2504     &  0.2500   \\
        &                     &  Bias    & -7.92e-04  &   5.60e-06  &   7.40e-04   &  -6.20e-06  &   -4.67e-04  &   2.68e-06  &   4.46e-04   &  -3.57e-06\\
        &                     &  MSE     & 6.32e-07   &   3.19e-11  &   5.49e-07   &   3.85e-11  &   2.19e-07   &   7.27e-12  &   2.00e-07   &  1.28e-11 \\
        &                     &  AVar    & 9.37e-10   &   8.78e-14  &   9.37e-10   &   8.78e-14  &   9.37e-10   &   8.78e-14  &   9.37e-10   &  8.78e-14 \\  \hline
        &   Second Component  & Parameters  & $\alpha_2$ & $\beta_2$ & $\gamma_2$ & $\delta_2$ & $\alpha_2$ & $\beta_2$ & $\gamma_2$ & $\delta_2$ \\
        &                     & True values & 1.5        & 0.5       & 1.75       & 0.75       & 1.5        & 0.5       & 1.75       & 0.75       \\ \hline
        &                     &  Average & 1.5000     &   0.5000    &   1.7500     &   0.7500    &   1.5000     &   0.5000    &   1.7500     &  0.7500    \\
        &                     &  Bias    & 1.94e-05   &  -1.60e-07  &   -8.18e-06  &   1.64e-08  &   1.46e-05   &  -1.19e-07  &   -3.32e-06  &   -2.62e-08\\
        &                     &  MSE     & 3.17e-09   &   2.88e-13  &   3.25e-09   &   2.91e-13  &   2.99e-09   &   2.73e-13  &   3.15e-09   &  2.89e-13  \\
        &                     &  AVar    & 2.95e-09   &   2.77e-13  &   2.95e-09   &   2.77e-13  &   2.95e-09   &   2.77e-13  &   2.95e-09   &  2.77e-13  \\ \hline
0.50    &    First Component  & Parameters  & $\alpha_1$ & $\beta_1$ & $\gamma_1$ & $\delta_1$ & $\alpha_1$ & $\beta_1$ & $\gamma_1$ & $\delta_1$ \\
        &                     & True values & 2.1        & 0.1       & 1.25       & 0.25       & 2.1        & 0.1       & 1.25       & 0.25       \\ \hline
        &                     &  Average & 2.0992     &   0.1000    &   1.2507     &   0.2500    &   2.0995     &   0.1000    &   1.2504     &  0.2500    \\
        &                     &  Bias    & -7.93e-04  &   5.61e-06  &   7.35e-04   &  -6.15e-06  &   -4.61e-04  &   2.63e-06  &   4.41e-04   &  -3.52e-06 \\
        &                     &  MSE     & 6.55e-07   &   3.38e-11  &   5.65e-07   &   4.01e-11  &   2.37e-07   &   9.07e-12  &   2.19e-07   &  1.47e-11  \\
        &                     &  AVar    & 2.34e-08   &   2.20e-12  &   2.34e-08   &   2.20e-12  &   2.34e-08   &   2.20e-12  &   2.34e-08   &  2.20e-12  \\ \hline
        &   Second Component  & Parameters & $\alpha_2$ & $\beta_2$ & $\gamma_2$ & $\delta_2$ & $\alpha_2$ & $\beta_2$ & $\gamma_2$ & $\delta_2$ \\
        &                     & True values & 1.5        & 0.5       & 1.75       & 0.75       & 1.5        & 0.5       & 1.75       & 0.75       \\ \hline
        &                     &  Average & 1.5000     &   0.5000    &   1.7500     &   0.7500    &   1.5000     &   0.5000    &   1.7500     &  0.7500    \\
        &                     &  Bias    & 1.84e-05   &  -1.28e-07  &   -1.52e-05  &   1.01e-07  &   1.36e-05   &  -8.67e-08  &   -1.06e-05  &   6.10e-08 \\
        &                     &  MSE     & 8.03e-08   &   7.24e-12  &   7.48e-08   &   6.90e-12  &   7.96e-08   &   7.16e-12  &   7.42e-08   &  6.86e-12  \\
        &                     &  AVar    & 7.38e-08   &   6.92e-12  &   7.38e-08   &   6.92e-12  &   7.38e-08   &   6.92e-12  &   7.38e-08   &  6.92e-12  \\ \hline
1.00    &    First Component  & Parameters  & $\alpha_1$ & $\beta_1$ & $\gamma_1$ & $\delta_1$ & $\alpha_1$ & $\beta_1$ & $\gamma_1$ & $\delta_1$ \\
        &                     & True values & 2.1        & 0.1       & 1.25       & 0.25       & 2.1        & 0.1       & 1.25       & 0.25       \\ \hline
        &                     &  Average & 2.0992     &   0.1000    &   1.2507     &   0.2500    &   2.0995     &   0.1000    &   1.2505     &  0.2500    \\
        &                     &  Bias    & -7.86e-04  &   5.52e-06  &   7.46e-04   &  -6.24e-06  &   -4.53e-04  &   2.53e-06  &   4.52e-04   &  -3.61e-06 \\
        &                     &  MSE     & 7.16e-07   &   3.95e-11  &   6.55e-07   &   4.79e-11  &   2.99e-07   &   1.49e-11  &   3.00e-07   &  2.17e-11  \\
        &                     &  AVar    & 9.37e-08   &   8.78e-12  &   9.37e-08   &   8.78e-12  &   9.37e-08   &   8.78e-12  &   9.37e-08   &  8.78e-12  \\ \hline
        &   Second Component  & Parameters  & $\alpha_2$ & $\beta_2$ & $\gamma_2$ & $\delta_2$ & $\alpha_2$ & $\beta_2$ & $\gamma_2$ & $\delta_2$ \\
        &                     & True values & 1.5        & 0.5       & 1.75       & 0.75       & 1.5        & 0.5       & 1.75       & 0.75       \\ \hline
        &                     &  Average & 1.5000     &   0.5000    &   1.7500     &   0.7500    &   1.5000     &   0.5000    &   1.7500     &  0.7500    \\
        &                     &  Bias    & -1.78e-05  &   1.53e-07  &   1.71e-05   &  -2.75e-07  &   -2.13e-05  &   1.78e-07  &   1.87e-05   &  -2.92e-07 \\
        &                     &  MSE     & 3.05e-07   &   2.86e-11  &   3.02e-07   &   2.78e-11  &   3.03e-07   &   2.83e-11  &   2.98e-07   &  2.76e-11  \\
        &                     &  AVar    & 2.95e-07   &   2.77e-11  &   2.95e-07   &   2.77e-11  &   2.95e-07   &   2.77e-11  &   2.95e-07   &  2.77e-11  \\ \hline
\end{tabular}}
\caption{Estimates of the parameters of model \eqref{multiple_comp_model} when M =  N = 100}
\label{table:8}
\end{table}

\subsection{Simulated Data Analysis}\label{sec:data_analysis}
We analyse a synthetic texture data using model \eqref{multiple_comp_model} to demonstrate  how the proposed parameter estimation methods work. The synthetic data is generated using the following model structure and parameters:
\begin{equation}
y(m,n) = \sum_{k=1}^{5}\{A_k^0 \cos(\alpha_k^0 m + \beta_k^0 m^2 + \gamma_k^0 n + \delta_k^0 n^2) + B_k^0 \sin(\alpha_k^0 m + \beta_k^0 m^2 + \gamma_k^0 n + \delta_k^0 n^2)\}
\end{equation}

\noindent The true parameter values are provided in Table \ref{table9}. The errors $X(m,n)$s are i.i.d.\ random variables with mean 0 and variance 100. Figure \ref{fig:true_signal} represents the original texture without any contamination and Figure \ref{fig:noisy_signal} represents the noisy texture. Our purpose is to extract the original gray-scale texture from the one which is contaminated.
  \begin{table}[]
 \centering
 \caption{True parameters values of the synthetic data.}
 \resizebox{0.9\textwidth}{!}{\begin{tabular}{cc|cc|cc|cc|cc|cc}
 \hline
$A_1^0$ & 6     &  $B_1^0$  &  6    &   $\alpha_1^0$  &  2.75   &  $\beta_1^0$ & 0.05    & $\gamma_1^0$  & 2.5  &   $\delta_1^0$ &    0.075  \\
$A_2^0$ &2      &  $B_2^0$  &    2    &   $\alpha_2^0$    &      1.75    & $\beta_2^0$ &  0.01   &  $\gamma_2^0$  &  1.5    &   $\delta_2^0$ &  0.025   \\

$A_3^0$ &1      &   $B_3^0$  &   1   &   $\alpha_3^0$   &      1.5     &  $\beta_3^0$ &   0.15     &  $\gamma_3^0$  &  2      &   $\delta_3^0$ &  0. 25  \\

$A_4^0$ &0.5     &   $B_4^0$  & 0.5   &   $\alpha_4^0$   &      1.75  & $\beta_4^0$ &  0.75    &  $\gamma_4^0$  &  2.75    &    $\delta_4^0$ & 0.275   \\

$A_5^0$ &0.1  &    $B_5^0$  &  0.1   &   $\alpha_5^0$     &    1.95 &  $\beta_5^0$  &   0.95      & $\gamma_5^0$  &   2.95      &   $\delta_5^0$ &  0.295    \\ \hline
\end{tabular}}
\label{table9}
\end{table}

\noindent We model the data using the proposed sequential procedure and the parameter estimates obtained using the sequential estimators are presented in Table \ref{table10}. From the obtained estimates, it can be inferred that the first four components are estimated satisfactorily but the last component is hardly detected. This makes sense as the amplitudes corresponding to the last component are very small.
The estimated texture is plotted in Figure \ref{fig:estimated_signal_LSE} and it is evident that the estimated texture and the original texture look extremely well matched. 
  \begin{table}[]
 \centering
 \caption{Estimates obtained using the sequential procedure for the synthetic data.}
 \resizebox{\textwidth}{!}{\begin{tabular}{cc|cc|cc|cc|cc|cc}
 \hline
$\hat{A}_1$ &  5.8773     &  $\hat{B}_1$  &  6.1384   &   $\hat{\alpha}_1$     &    2.7499    &  $\hat{\beta}_1$ &  0.0500     &  $\hat{\gamma}_1$  &   2.5004    &    $\hat{\delta}_1$ &    0.0749  \\
$\hat{A}_2$ &  2.2789     &  $\hat{B}_2$  &  1.7718   &   $\hat{\alpha}_2$     &    1.7492    &  $\hat{\beta}_2$ &  0.0100     &  $\hat{\gamma}_2$  &   1.4988    &    $\hat{\delta}_2$ &    0.0250   \\
$\hat{A}_3$ &  1.0856     &  $\hat{B}_3$  &  0.9090   &   $\hat{\alpha}_3$     &    1.4999    &  $\hat{\beta}_3$ &  0.1499     &  $\hat{\gamma}_3$  &   1.9979    &    $\hat{\delta}_3$ &    0.2500 \\
$\hat{A}_4$ &  0.4828     &  $\hat{B}_4$  &  0.5418   &   $\hat{\alpha}_4$     &    1.7482    &  $\hat{\beta}_4$ &  0.7500     &  $\hat{\gamma}_4$  &   2.7547    &    $\hat{\delta}_4$ &    0.2749   \\
$\hat{A}_5$ &  0.0251     &  $\hat{B}_5$  & -0.0106   &   $\hat{\alpha}_5$     &    2.2254    &  $\hat{\beta}_5$ &  1.1450     &  $\hat{\gamma}_5$  &   3.2173    &    $\hat{\delta}_5$ &    0.5945   \\
 \hline
\end{tabular}}
\label{table10}
\end{table}
\begin{figure}[H]
\centering
\includegraphics[scale=0.5]{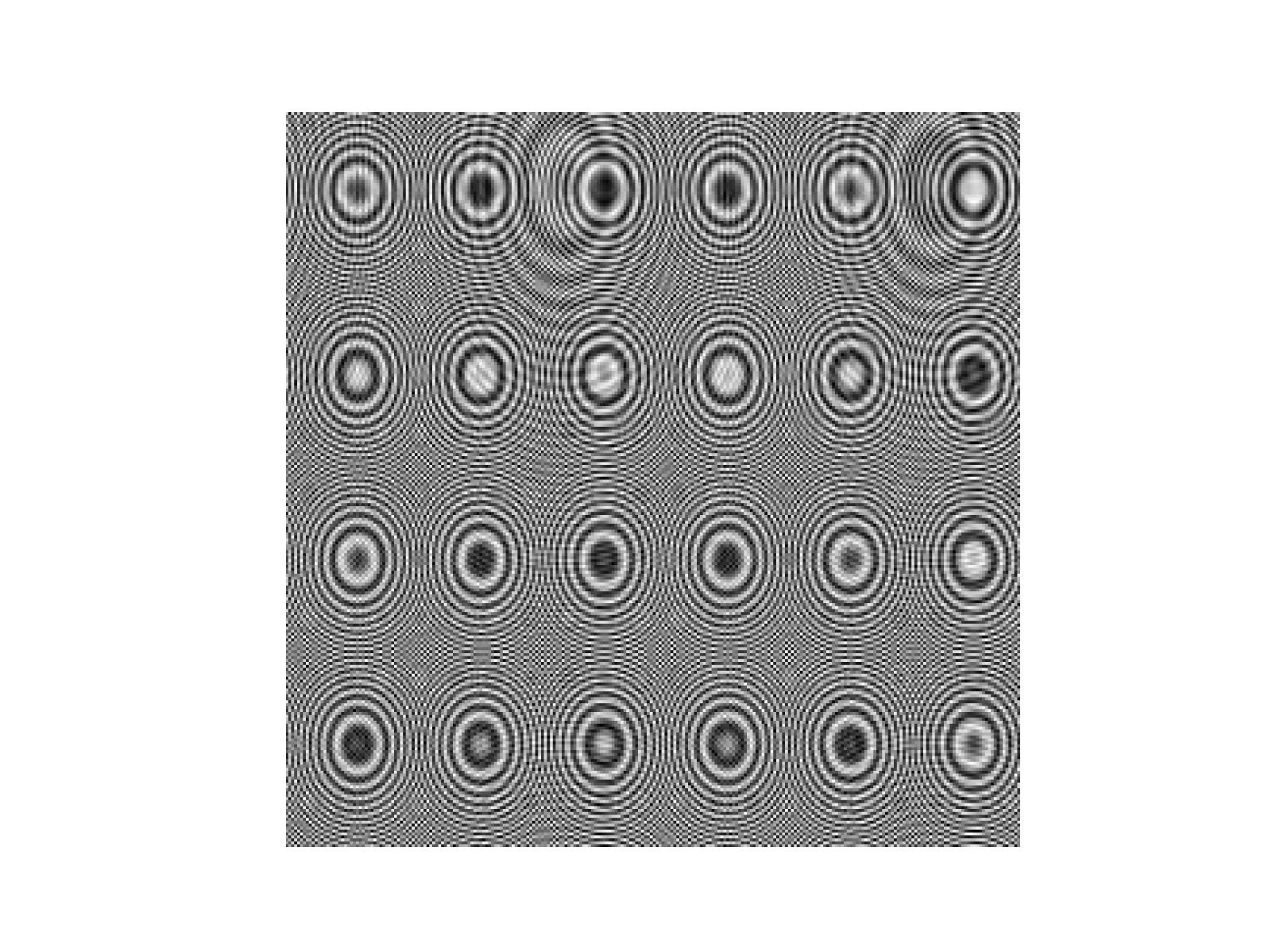}
\caption{Estimated texture for the synthetic data.}
\label{fig:estimated_signal_LSE}
\end{figure}

\section{Concluding Remarks}\label{sec:Conclusion}
In this paper, we have considered the estimation of unknown parameters of a 2-D chirp model under the assumption of i.i.d.\ additive errors. The main idea is to reduce the computational complexity involved in finding the LSEs of these parameters. The proposed estimators minimise the computations to a great extent and are observed to be strongly consistent and asymptotically equivalent to the LSEs. For a 2-D chirp model with $p$ number of components, we have proposed a sequential procedure which reduces the problem of estimation of the parameters to solving $p$ number of 2-D optimisation problems. Moreover, the propounded sequential estimators are observed to be strongly consistent and asymptotically equivalent to the usual LSEs.\\

An alternative method to estimate the non-linear parameters of a one-component 2-D chirp model is to maximize the following periodogram-type functions:
\begin{equation}\begin{split}\label{periodogram-type_function_alpha_beta}
I^{(1)}_{MN}(\alpha, \beta) & = \frac{2}{MN}\sum\limits_{{n_0} = 1}^{N}\textit{\textbf{Y}}^{\top}_{n_0}\textbf{Z}_M(\alpha, \beta)\textbf{Z}_M(\alpha, \beta)^{\top}\textit{\textbf{Y}}_{n_0}
\end{split}\end{equation}
\begin{equation}\begin{split}\label{periodogram-type_function_gamma_delta}
I^{(2)}_{MN}(\gamma, \delta) & = \frac{2}{MN}\sum\limits_{{m_0} = 1}^{M}\textit{\textbf{Y}}^{\top}_{m_0}\textbf{Z}_N(\gamma, \delta)\textbf{Z}_N(\gamma, \delta)^{\top}\textit{\textbf{Y}}_{m_0}
\end{split}\end{equation} 
with respect to $\alpha$, $\beta$ and $\gamma$, $\delta$ respectively. 
These periodogram-type functions are constructed in the same way as the reduced error sum of squares functions defined in \eqref{reduced_ess_alpha_beta} and \eqref{reduced_ess_gamma_delta} with the same idea to reduce 2-D chirp model to a number of 1-D chirp models with same frequency and frequency rate parameters as the original 2-D model but with different amplitudes. \\

Using some number theoretic results established by Lahiri \cite{2015}, it is easy to show that the following relationship exists between the reduced error sum of squares and the periodogram-type functions:

\begin{equation}\label{R_I_relationship}
\frac{1}{N}R^{(1)}_{MN}(\alpha, \beta) = \frac{1}{N}\sum_{{n_0} = 1}^{N} \textit{\textbf{Y}}^{\top}_{n_0}\textit{\textbf{Y}}_{n_0} - I^{(1)}_{MN}(\alpha, \beta)  + o(1)
\end{equation}
Here, a function $f$ is $o(1)$, if $f \to 0$ almost surely as $M \to \infty$. A similar relation can be seen between $R^{(2)}_{MN}(\gamma, \delta)$ and $I^{(2)}_{MN}(\gamma, \delta)$. Thus, replacing the functions $R^{(1)}_{MN}(\alpha, \beta)$ and $R^{(2)}_{MN}(\gamma, \delta)$ by $I^{(1)}_{MN}(\alpha, \beta)$ and $I^{(2)}_{MN}(\gamma, \delta)$ respectively is plausible as its effect on the estimators will be inconsequential. However, this replacement simplifies the estimation process to a great extent as the evaluation of periodogram-type functions does not involve matrix inversion. \\

This relationship is analogous to the one that was first proposed by Walker \cite{1971} for the sinusoidal model and later Grover et al. \cite{2018, 2018_2} extended the same for 1-D and 2-D chirp models. The estimators obtained by maximising a periodogram function \cite{1971} or a periodogram-type function \cite{2018, 2018_2} are called the approximate least squares estimators (ALSEs). In fact, Grover et al. \cite{2018, 2018_2} showed that the ALSEs are strongly consistent and asymptotically equivalent to the corresponding LSEs. Furthermore, they showed that the ALSEs have two distinctive and noteworthy aspects$-$ \textit{(a)} their consistency is obtained under slightly less restrictive assumptions on the linear parameters than those required for the LSEs and \textit{(b)} their computation is faster as compared to the LSEs due to absence of a matrix inversion in the former case. Therefore, it will be interesting to investigate the behaviour of the estimators obtained by maximising functions \eqref{periodogram-type_function_alpha_beta} and \eqref{periodogram-type_function_gamma_delta} and to assess their computational performance as compared to the estimators we proposed in this paper. \\

The numerical experiments$-$ the simulations and the data analysis, show that the proposed estimation technique provides as accurate results as the least squares estimation method with the additional advantage of being computationally more efficient. Thus to summarise, the proposed estimators seem to be the method of choice as their performance is satisfactory and as efficient as the LSEs, both numerically and analytically.

\section*{Appendix A}\label{appendix:A}
\noindent Henceforth, we will denote $\boldsymbol{\theta}(n_0) = (A(n_0), B(n_0), \alpha, \beta)$ as the parameter vector and $\boldsymbol{\theta}^0(n_0) = (A^0(n_0), B^0(n_0), \alpha^0, \beta^0)$ as the true parameter vector of the 1-D chirp model \eqref{model_n=n0_one_comp}. \\
To prove Theorem \ref{theorem:consistency_alpha_beta_one_comp_LSE}, we need the following lemma:
\begin{lemma}\label{lemma_consistency_LSEs_alpha_beta}
Consider the set $S_c$ = $\{(\alpha, \beta) : |\alpha - \alpha^0| \geqslant c \textmd{ or } |\beta - \beta^0| \geqslant c\}$. If for any c $>$ 0, \\
\begin{equation}\label{condition_for_consistency_LSE_alpha_beta}
\liminf \inf\limits_{(\alpha, \beta) \in S_c} \frac{1}{M N}\bigg[R^{(1)}_{MN}(\alpha, \beta) - R^{(1)}_{MN}(\alpha^0, \beta^0)\bigg] > 0 \ a.s.
\end{equation}
then, $\hat{\alpha}$ $\rightarrow$ $\alpha^0$  and $\hat{\beta}$ $\rightarrow$ $\beta^0$ almost surely as $M \rightarrow \infty$.
\end{lemma}
\noindent \begin{proof}
This proof follows along the same lines as that of Lemma 1 of Wu \cite{1981}.\\
\end{proof}

\noindent \textbf{Proof of Theorem \ref{theorem:consistency_alpha_beta_one_comp_LSE}:} Let us consider the following:
\begin{flalign*}
& \liminf \inf\limits_{(\alpha, \beta) \in S_c} \frac{1}{M N}\bigg[R^{(1)}_{MN}(\alpha, \beta) - R^{(1)}_{MN}(\alpha^0, \beta^0)\bigg]& \\
& = \liminf \inf\limits_{(\alpha, \beta) \in S_c} \frac{1}{M N}\bigg[\sum_{{n_0} = 1}^{N}R_{M}(\alpha, \beta, n_0) - \sum_{{n_0} = 1}^{N} R_{M}(\alpha^0, \beta^0, n_0)\bigg]& \\
& = \liminf \inf\limits_{(\alpha, \beta) \in S_c} \frac{1}{M N}\bigg[\sum_{{n_0} = 1}^{N}Q_{M}(\hat{A}(n_0), \hat{B}(n_0), \alpha, \beta) - \sum_{{n_0} = 1}^{N} Q_{M}(\hat{A}(n_0), \hat{B}(n_0), \alpha^0, \beta^0)\bigg] & \\
& \geqslant \liminf \inf\limits_{(\alpha, \beta) \in S_c} \frac{1}{M N}\bigg[\sum_{{n_0} = 1}^{N}Q_{M}(\hat{A}(n_0), \hat{B}(n_0), \alpha, \beta) - \sum_{{n_0} = 1}^{N} Q_{M}(A^0(n_0) , B^0(n_0), \alpha^0, \beta^0)\bigg] & \\
& \geqslant \liminf \inf\limits_{ \boldsymbol{\theta}(n_0) \in M_c^{n_0}} \frac{1}{M N}\bigg[\sum_{{n_0} = 1}^{N}Q_{M}(A(n_0), B(n_0), \alpha, \beta) - \sum_{{n_0} = 1}^{N} Q_{M}(A^0(n_0) , B^0(n_0), \alpha^0, \beta^0)\bigg] & \\
&\geqslant \frac{1}{N}\sum_{n_0 = 1}^{N} \liminf \inf\limits_{\boldsymbol{\theta}(n_0) \in M_c^{n_0}} \frac{1}{M}\bigg[Q_M(\boldsymbol{\theta}(n_0)) - Q_M(\boldsymbol{\theta}^0(n_0))\bigg] > 0. &
\end{flalign*}
This follows from the proof of Theorem 1 of Kundu and Nandi \cite{2008}. Here,  $Q_{M}(A(n_0), B(n_0), \alpha, \beta)  =  \textit{\textbf{Y}}^{\top}_{n_0}(\textit{\textbf{I}} - \textbf{Z}_M(\alpha, \beta)(\textbf{Z}_M(\alpha, \beta)^{\top}\textbf{Z}_M(\alpha, \beta))^{-1}\textbf{Z}_M(\alpha, \beta)^{\top})\textit{\textbf{Y}}_{n_0}$. Also note that the set $M_c^{n_0}$ = $\{\boldsymbol{\theta}(n_0) : |A(n_0) - A^0(n_0)| \geqslant c \textmd{ or } |B(n_0) - B^0(n_0)| \geqslant c \textmd{ or } |\alpha - \alpha^0| \geqslant c \textmd{ or } |\beta - \beta^0| \geqslant c\}$ which implies $S_c \subset M_c^{n_0}$, for all $n_0 \in \{1, \ldots, N\}$.
Thus, using Lemma \ref{lemma_consistency_LSEs_alpha_beta}, $\hat{\alpha} \xrightarrow{a.s.} \alpha^0$ and $\hat{\beta} \xrightarrow{a.s.} \beta^0$. \\ \qed

\noindent \textbf{Proof of Theorem \ref{theorem:asymp_dist_one_comp_LSE_alpha_beta}:} Let us denote $\boldsymbol{\xi} = (\alpha, \beta)$ and $\hat{\boldsymbol{\xi}} = (\hat{\alpha}, \hat{\beta})$, the estimator of $\boldsymbol{\xi}^0 = (\alpha^0, \beta^0)$ obtained by minimising the function $R_{MN}^{(1)}(\boldsymbol{\xi}) = R_{MN}^{(1)}(\alpha, \beta)$ defined in \eqref{reduced_ess_alpha_beta}. \\  
\noindent Using multivariate Taylor series, we expand the $1 \times 2$ first derivative vector $\textit{\textbf{R}}_{MN}^{(1)'}(\hat{\boldsymbol{\xi}})$ of the function $R_{MN}^{(1)}(\boldsymbol{\xi})$, around the point $\boldsymbol{\xi}^0$ as follows:
\begin{equation*}
\textit{\textbf{R}}_{MN}^{(1)'}(\hat{\boldsymbol{\xi}}) - \textit{\textbf{R}}_{MN}^{(1)'}(\xi^0) = (\hat{\boldsymbol{\xi}} - \boldsymbol{\xi}^0)\textit{\textbf{R}}_{MN}^{(1)''}(\bar{\boldsymbol{\xi}}),
\end{equation*}
where $\bar{\boldsymbol{\xi}}$ is a point between $\hat{\boldsymbol{\xi}}$ and $\boldsymbol{\xi}^0$ and $\textit{\textbf{R}}_{MN}^{(1)''}(\bar{\boldsymbol{\xi}})$ is the $2 \times 2$ second derivative matrix of the function $R_{MN}^{(1)}(\boldsymbol{\xi})$ at the point $\bar{\boldsymbol{\xi}}$. Since $\hat{\boldsymbol{\xi}}$ minimises the function $R_{MN}^{(1)}(\boldsymbol{\xi})$, $\textit{\textbf{R}}_{MN}^{(1)'}(\hat{\boldsymbol{\xi}}) = 0$. Thus, we have
\begin{equation*}
(\hat{\boldsymbol{\xi}} - \boldsymbol{\xi}^0) = - \textit{\textbf{R}}_{MN}^{(1)'}(\boldsymbol{\xi}^0)[\textit{\textbf{R}}_{MN}^{(1)''}(\bar{\boldsymbol{\xi}})]^{-1}. 
\end{equation*}
Multiplying both sides by the diagonal matrix $\textit{\textbf{D}}_1^{-1} = \textnormal{diag}(M^{\frac{-3}{2}}N^{\frac{-1}{2}}, M^{\frac{-5}{2}}N^{\frac{-1}{2}})$, we get:
\begin{equation}\label{taylor_series_R_MN'}
(\hat{\boldsymbol{\xi}} - \boldsymbol{\xi}^0)\textit{\textbf{D}}_1^{-1} = - \textit{\textbf{R}}_{MN}^{(1)'}(\boldsymbol{\xi}^0)\textit{\textbf{D}}_1[\textit{\textbf{D}}_1R_{MN}^{(1)''}(\bar{\boldsymbol{\xi}})\textit{\textbf{D}}_1]^{-1}. 
\end{equation}
Consider the vector, $$\textit{\textbf{R}}_{MN}^{(1)'}(\boldsymbol{\xi}^0)\textit{\textbf{D}}_1 = \begin{bmatrix}
\frac{1}{M^{3/2}N^{1/2}}\frac{\partial R_{MN}^{(1)}(\boldsymbol{\xi}^0)}{\partial \alpha} & \frac{1}{M^{3/2}N^{1/2}}\frac{\partial R_{MN}^{(1)}(\boldsymbol{\xi}^0)}{\partial \beta}
\end{bmatrix}.$$ 
On computing the elements of this vector and using preliminary result \eqref{prelim_LSE_first_derivative_one_comp} (see Section \ref{sec:one_comp_1D}) and the definition of the function:
$$R_{MN}^{(1)}(\alpha, \beta) = \sum\limits_{{n_0} = 1}^{N}R_M(\alpha, \beta, n_0)$$
we obtain the following result:
\begin{equation}\label{R_MN'_distribution}
-\textit{\textbf{R}}_{MN}^{(1)'}(\boldsymbol{\xi}^0)\textit{\textbf{D}}_1 \xrightarrow{d} \boldsymbol{\mathcal{N}}_2(\textbf{0}, 2\sigma^2\boldsymbol{\Sigma}) \textmd{ as } M \rightarrow \infty.
\end{equation}
Since $\hat{\boldsymbol{\xi}} \xrightarrow{a.s.} \boldsymbol{\xi}^0$, and as each element of the matrix $\textit{\textbf{R}}_{MN}^{(1)''}(\boldsymbol{\xi})$  is a continuous function of $\boldsymbol{\xi}$, we have
\begin{equation*}
\lim_{M \rightarrow \infty}\textit{\textbf{D}}_1\textbf{\textit{R}}_{MN}^{(1)''}(\bar{\boldsymbol{\xi}})\textit{\textbf{D}}_1 = \lim_{M \rightarrow \infty}\textit{\textbf{D}}_1\textit{\textbf{R}}_{MN}^{(1)''}(\boldsymbol{\xi}^0)\textit{\textbf{D}}_1.
\end{equation*} 
Now using preliminary result \eqref{prelim_LSE_second_derivative_one_comp} (see Section \ref{sec:one_comp_1D}), it can be seen that:
\begin{equation}\label{R_MN''_limit}
\lim_{M \rightarrow \infty}\textit{\textbf{D}}_1\textbf{\textit{R}}_{MN}^{(1)''}(\boldsymbol{\xi}^0)\textit{\textbf{D}}_1 \rightarrow \boldsymbol{\Sigma}^{-1}.
\end{equation} 
On combining \eqref{taylor_series_R_MN'}, \eqref{R_MN'_distribution} and \eqref{R_MN''_limit}, we have the desired result.

\section*{Appendix B}\label{appendix:B}
To prove Theorem \ref{theorem:consistency_alphak_betak_multiple_comp_LSE}, we need the following lemmas:
\begin{lemma}\label{lemma_consistency_LSEs_alpha1_beta1}
Consider the set $S_c^{1}$ = $\{(\alpha, \beta) : |\alpha - \alpha_1^0| \geqslant c \textmd{ or } |\beta - \beta_1^0| \geqslant c\}$.If for any c $>$ 0, \\
\begin{equation}\label{condition_for_consistency_LSE_alpha1_beta1}
\liminf \inf\limits_{(\alpha, \beta) \in S_c^{1}} \frac{1}{M N}[R^{(1)}_{1,MN}(\alpha, \beta) - R^{(1)}_{1,MN}(\alpha_1^0, \beta_1^0)] > 0 \ a.s.
\end{equation}
then, $\hat{\alpha}_1$ $\rightarrow$ $\alpha_1^0$ and $\hat{\beta}_1$ $\rightarrow$ $\beta_1^0$ almost surely as $M \rightarrow \infty$.
\end{lemma}
\noindent \begin{proof}
This proof follows along the same lines as proof of Lemma \ref{lemma_consistency_LSEs_alpha_beta}. \\
\end{proof}

\begin{lemma}\label{lemma:convergence_LSE_alpha1_beta1}
If assumptions \ref{assump:1}, \ref{assump:3} and \ref{assump:P4} are satisfied then:
\begin{equation*}\begin{split}
& M(\hat{\alpha}_1 - \alpha_1^0) \xrightarrow{a.s.} 0,\\
& M^2(\hat{\beta}_1 - \beta_1^0) \xrightarrow{a.s.} 0.
\end{split}\end{equation*}
\end{lemma}
\noindent \begin{proof}
 Let us denote $\textit{\textbf{R}}_{1, MN}^{(1)'}(\boldsymbol{\xi})$ as the $1 \times 2$ first derivative vector and $\textit{\textbf{R}}_{1, MN}^{(1)''}(\boldsymbol{\xi})$ as the $2 \times 2$ second derivative matrix of the function $R_{1, MN}^{(1)}(\boldsymbol{\xi})$. Using multivariate Taylor series expansion, we expand the function $\textit{\textbf{R}}_{1, MN}^{(1)'}(\hat{\boldsymbol{\xi}}_1)$ around the point $\boldsymbol{\xi}_1^0$ as follows:
\begin{equation*}
\textit{\textbf{R}}_{1, MN}^{(1)'}(\hat{\boldsymbol{\xi}}_1) - \textit{\textbf{R}}_{1, MN}^{(1)'}(\boldsymbol{\xi}_1^0) = (\hat{\boldsymbol{\xi}}_1 - \boldsymbol{\xi}_1^0)\textit{\textbf{R}}_{1, MN}^{(1)''}(\bar{\boldsymbol{\xi}}_1)
\end{equation*}
where $\bar{\boldsymbol{\xi}}_1$ is a point between $\hat{\boldsymbol{\xi}}_1$ and $\boldsymbol{\xi}_1^0$. Note that $\textit{\textbf{R}}_{1, MN}^{(1)'}(\hat{\boldsymbol{\xi}}_1) = 0$. Thus, we have:
\begin{equation}\label{taylor_series_R_1_MN'_original}
(\hat{\boldsymbol{\xi}}_1 - \boldsymbol{\xi}_1^0) = - \textit{\textbf{R}}_{1, MN}^{(1)'}(\boldsymbol{\xi}_1^0)[\textit{\textbf{R}}_{1, MN}^{(1)''}(\bar{\boldsymbol{\xi}}_1)]^{-1}.
\end{equation}
Multiplying both sides by $\frac{1}{\sqrt{MN}}\textit{\textbf{D}}_1^{-1}$, we get:
\begin{equation}\label{taylor_series_R_1_MN'}
(\hat{\boldsymbol{\xi}}_1 - \boldsymbol{\xi}_1^0)(\sqrt{MN}\textit{\textbf{D}}_1)^{-1} = - \frac{1}{\sqrt{MN}}\textit{\textbf{R}}_{1, MN}^{(1)'}(\boldsymbol{\xi}_1^0)\textit{\textbf{D}}_1[\textit{\textbf{D}}_1\textit{\textbf{R}}_{1, MN}^{(1)''}(\bar{\boldsymbol{\xi}}_1)\textit{\textbf{D}}_1]^{-1}.
\end{equation}
Since each of the elements of the matrix $\textit{\textbf{R}}_{1, MN}^{(1)''}(\boldsymbol{\xi})$ is a continuous function of $\boldsymbol{\xi},$
$$ \lim_{M \rightarrow \infty} \textit{\textbf{D}}_1\textit{\textbf{R}}_{1, MN}^{(1)''}(\bar{\boldsymbol{\xi}}_1)\textit{\textbf{D}}_1 = \lim_{M \rightarrow \infty} \textit{\textbf{D}}_1\textit{\textbf{R}}_{1, MN}^{(1)''}(\boldsymbol{\xi}_1^0)\textit{\textbf{D}}_1. $$
By definition, 
\begin{equation}\label{definition_R_1_MN}
R_{1, MN}^{(1)}(\boldsymbol{\xi}) = \sum_{{n_0} = 1}^{N} R_{1,M}(\boldsymbol{\xi}, n_0).
\end{equation}
Using this and the preliminary result \eqref{prelim_LSE_first_derivative_multiple_comp_1} and \eqref{prelim_LSE_second_derivative_multiple_comp} (see Section \ref{sec:multiple_comp_1D}), it can be seen that:
\begin{align}
- \frac{1}{\sqrt{MN}}\textit{\textbf{R}}_{1, MN}^{(1)'}(\boldsymbol{\xi}_1^0)\textit{\textbf{D}}_1 \xrightarrow{a.s.} 0 \textmd{ as } M \rightarrow \infty. \label{limit_R_1_MN'}\\
\textit{\textbf{D}}_1\textit{\textbf{R}}_{1, MN}^{(1)''}(\boldsymbol{\xi}_1^0)\textit{\textbf{D}}_1 \xrightarrow{a.s.} \boldsymbol{\Sigma}_1^{-1} \textmd{ as } M \rightarrow \infty.\label{limit_R_1_MN''}
\end{align}
On combining \eqref{taylor_series_R_1_MN'}, \eqref{limit_R_1_MN'} and \eqref{limit_R_1_MN''}, we have the desired result.\\
\end{proof}

\noindent \textbf{Proof of Theorem \ref{theorem:consistency_alphak_betak_multiple_comp_LSE}:} Consider the left hand side of \eqref{condition_for_consistency_LSE_alpha1_beta1}, that is,
\begin{flalign*}
& \liminf \inf\limits_{(\alpha, \beta) \in S_c^{1}} \frac{1}{M N}\bigg[R^{(1)}_{1,MN}(\alpha, \beta) - R^{(1)}_{1,MN}(\alpha_1^0, \beta_1^0)\bigg]& \\
& = \liminf \inf\limits_{(\alpha, \beta) \in S_c^{1}} \frac{1}{M N}\bigg[\sum_{{n_0} = 1}^{N}Q_{1,M}(\hat{A}_1(n_0), \hat{B}_1(n_0), \alpha, \beta) - \sum_{{n_0} = 1}^{N} Q_{1,M}(\hat{A}_1(n_0), \hat{B}_1(n_0), \alpha_1^0, \beta_1^0)\bigg] & \\
& \geqslant \liminf \inf\limits_{(\alpha, \beta) \in S_c^{1}} \frac{1}{M N}\bigg[\sum_{{n_0} = 1}^{N}Q_{1,M}(\hat{A}_1(n_0), \hat{B}_1(n_0), \alpha, \beta) - \sum_{{n_0} = 1}^{N} Q_{1,M}(A_1^0(n_0) , B_1^0(n_0), \alpha_1^0, \beta_1^0)\bigg] & \\
& \geqslant \liminf \inf\limits_{ \boldsymbol{\theta}_1(n_0) \in M_c^{1,n_0}} \frac{1}{M N}\bigg[\sum_{{n_0} = 1}^{N}Q_{1,M}(A_1(n_0), B_1(n_0), \alpha, \beta) - \sum_{{n_0} = 1}^{N} Q_{1,M}(A_1^0(n_0) , B_1^0(n_0), \alpha_1^0, \beta_1^0)\bigg] & \\
&\geqslant \frac{1}{N}\sum_{n_0 = 1}^{N} \liminf \inf\limits_{\boldsymbol{\theta}_1(n_0) \in M_c^{1,n_0}} \frac{1}{M}\bigg[Q_{1,M}(\boldsymbol{\theta}_1(n_0)) - Q_{1,M}(\boldsymbol{\theta}_1^0(n_0))\bigg] > 0. &
\end{flalign*}
Here,  $Q_{1,M}(A(n_0), B(n_0), \alpha, \beta)  =  \textit{\textbf{Y}}^{\top}_{n_0}(\textit{\textbf{I}} - \textbf{Z}_M(\alpha, \beta)(\textbf{Z}_M(\alpha, \beta)^{\top}\textbf{Z}_M(\alpha, \beta))^{-1}\textbf{Z}_M(\alpha, \beta)^{\top})\textit{\textbf{Y}}_{n_0}$ and $M_c^{1,n_0}$ can be obtained by replacing $\alpha^0$ and $\beta^0$ by $\alpha_1^0$ and $\beta_1^0$ respectively, in the set $M_c^{n_0}$ defined in Lemma \ref{lemma_consistency_LSEs_alpha_beta}.  The last step follows from the proof of Theorem 2.4.1 of Lahiri \cite{2015}.
Thus, using Lemma \ref{lemma_consistency_LSEs_alpha1_beta1}, $\hat{\alpha}_1 \xrightarrow{a.s.} \alpha_1^0$ and $\hat{\beta}_1 \xrightarrow{a.s.} \beta_1^0$ as $M \rightarrow\infty$. \\
Following similar arguments, one can obtain the consistency of $\hat{\gamma}_1$ and $\hat{\delta}_1$ as $N \rightarrow \infty$. 
Also,
\begin{equation*}\begin{split}
& N(\hat{\gamma}_1 - \gamma_1^0) \xrightarrow{a.s.} 0,\\
& N^2(\hat{\delta}_1 - \delta_1^0) \xrightarrow{a.s.} 0.
\end{split}\end{equation*}
The proof of the above equations follows along the same lines as the proof of Lemma \ref{lemma:convergence_LSE_alpha1_beta1}. 
From Theorem \ref{theorem:limit_A_p+1_B_p+1}, it follows that as min$\{M, N\} \rightarrow \infty$:
\begin{equation*}\begin{split}
& (\hat{A}_1 - A_1^0) \xrightarrow{a.s.} 0,\\
& (\hat{B}_1 - B_1^0) \xrightarrow{a.s.} 0.
\end{split}\end{equation*}
Thus, we have the following relationship between the first component of model \eqref{multiple_comp_model} and its estimate:
\begin{equation}\begin{split}\label{relationship_first_comp_true_estimate}
& \hat{A}_1 \cos(\hat{\alpha}_1 m + \hat{\beta}_1 m^2 + \hat{\gamma}_1 n + \hat{\delta}_1 n^2) + \hat{B}_1 \sin(\hat{\alpha}_1 m + \hat{\beta}_1 m^2 + \hat{\gamma}_1 n + \hat{\delta}_1 n^2) = \\
& \qquad \qquad \qquad \qquad \qquad \qquad A_1^0 \cos(\alpha_1^0 m + \beta_1^0 m^2 + \gamma_1^0 n + \delta_1^0 n^2) + B_1^0 \sin(\alpha_1^0 m + \beta_1^0 m^2 + \gamma_1^0 n + \delta_1^0 n^2) + o(1).
\end{split}\end{equation}
Here  a function $g$ is $o(1)$, if $g \to 0$ almost surely as min$\{M, N\} \to \infty$.\\

\noindent Using \eqref{relationship_first_comp_true_estimate} and following the same arguments as above for the consistency of $\hat{\alpha}_1$, $\hat{\beta}_1$, $\hat{\gamma}_1$ and $\hat{\delta}_1$, we can show that, $\hat{\alpha}_2$, $\hat{\beta}_2$, $\hat{\gamma}_2$ and $\hat{\delta}_2$ are strongly consistent estimators of  $\alpha_2^0$, $\beta_2^0$, $\gamma_2^0$ and $\delta_2^0$ respectively. And the same can be extended for $k \leqslant p.$ Hence, the result.\\ \qed

\noindent  \textbf{Proof of Theorem \ref{theorem:limit_A_p+1_B_p+1}:} We will consider the following two cases that will cover both the scenarios$-$ underestimation as well as overestimation of the number of components:
\begin{itemize}
\item \underline{Case 1:} When $k = 1$:
\begin{equation}\begin{split}\label{linear_estimates_1}
\begin{bmatrix}
\hat{A}_1 \\ \hat{B}_1
\end{bmatrix} = [\textit{\textbf{W}}(\hat{\alpha}_1, \hat{\beta}_1, \hat{\gamma}_1, \hat{\delta}_1)^{\top}\textit{\textbf{W}}(\hat{\alpha}_1, \hat{\beta}_1, \hat{\gamma}_1, \hat{\delta}_1)]^{-1}\textit{\textbf{W}}(\hat{\alpha}_1, \hat{\beta}_1, \hat{\gamma}_1, \hat{\delta}_1)^{\top} \textit{\textbf{Y}}
\end{split}\end{equation}
Using Lemma 1 of Lahiri et al. \cite{2015}, it can be seen that:
\begin{equation*}
\frac{1}{MN}[\textit{\textbf{W}}(\hat{\alpha}_1, \hat{\beta}_1, \hat{\gamma}_1, \hat{\delta}_1)^{\top}\textit{\textbf{W}}(\hat{\alpha}_1, \hat{\beta}_1, \hat{\gamma}_1, \hat{\delta}_1)] \rightarrow \frac{1}{2}\textit{\textbf{I}}_{2 \times 2} \textmd{ as } \min\{M, N\} \rightarrow \infty.
\end{equation*}
Substituting this result in \eqref{linear_estimates_1}, we get:
\begin{equation*}\begin{split}
\begin{bmatrix}
\hat{A}_1 \\ \hat{B}_1
\end{bmatrix} & = \frac{2}{MN}\textit{\textbf{W}}(\hat{\alpha}_1, \hat{\beta}_1, \hat{\gamma}_1, \hat{\delta}_1)^{\top} \textit{\textbf{Y}} + o(1)\\
& = \begin{bmatrix}
\frac{2}{MN}\sum\limits_{n=1}^{N}\sum\limits_{m=1}^{M}y(m,n)\cos(\hat{\alpha}_1 m + \hat{\beta}_1 m^2 + \hat{\gamma}_1 n + \hat{\delta}_1 n^2) + o(1) \\
\frac{2}{MN}\sum\limits_{n=1}^{N}\sum\limits_{m=1}^{M}y(m,n)\sin(\hat{\alpha}_1 m + \hat{\beta}_1 m^2 + \hat{\gamma}_1 n + \hat{\delta}_1 n^2) + o(1)
\end{bmatrix}.
\end{split}\end{equation*}
Now consider the estimate $\hat{A}_1$. Using multivariate Taylor series, we expand the function $\cos(\hat{\alpha}_1 m + \hat{\beta}_1 m^2 + \hat{\gamma}_1 n + \hat{\delta}_1 n^2)$ around the point $(\alpha_1^0, \beta_1^0, \gamma_1^0, \delta_1^0)$ and we obtain:
\begin{equation*}\begin{split}
\hat{A}_1 & = \frac{2}{MN}y(m,n)\bigg\{\cos(\alpha_1^0 m + \beta_1^0 m^2 + \gamma_1^0 n + \delta_1^0 n^2) - m (\hat{\alpha}_1 - \alpha_1^0)\sin(\alpha_1^0 m + \beta_1^0 m^2 + \gamma_1^0 n + \delta_1^0 n^2) \\
& - m^2(\hat{\beta}_1 - \beta_1^0)\sin(\alpha_1^0 m + \beta_1^0 m^2 + \gamma_1^0 n + \delta_1^0 n^2) - n (\hat{\gamma}_1 - \gamma_1^0)\sin(\alpha_1^0 m + \beta_1^0 m^2 + \gamma_1^0 n + \delta_1^0 n^2) \\
& - n^2(\hat{\delta}_1 - \delta_1^0)\sin(\alpha_1^0 m + \beta_1^0 m^2 + \gamma_1^0 n + \delta_1^0 n^2) \bigg\} \\
& \rightarrow 2 \times \frac{A_1^0}{2} = A_1^0 \textmd{ almost surely as } \min\{M, N\} \rightarrow \infty,
\end{split}\end{equation*}
using \eqref{multiple_comp_model} and Lemma 1 and Lemma 2 of Lahiri et al. \cite{2015}. Similarly, it can be shown that $\hat{B}_1 \rightarrow B_1^0$ almost surely as $\min\{M, N\} \rightarrow \infty$. \\ \\
For the second component linear parameter estimates, consider:
\begin{equation*}\begin{split}
\begin{bmatrix}
\hat{A}_2 \\ \hat{B}_2
\end{bmatrix} = \begin{bmatrix}
\frac{2}{MN}\sum\limits_{n=1}^{N}\sum\limits_{m=1}^{M}y_1(m,n)\cos(\hat{\alpha}_2 m + \hat{\beta}_2 m^2 + \hat{\gamma}_2 n + \hat{\delta}_2 n^2) + o(1)\\
\frac{2}{MN}\sum\limits_{n=1}^{N}\sum\limits_{m=1}^{M}y_1(m,n)\sin(\hat{\alpha}_2 m + \hat{\beta}_2 m^2 + \hat{\gamma}_2 n + \hat{\delta}_2 n^2) + o(1)
\end{bmatrix}.
\end{split}\end{equation*}
Here, $y_1(m,n)$ is the data obtained at the second stage after eliminating the effect of the first component from the original data as defined in \eqref{second_stage_data}. Using the relationship \eqref{relationship_first_comp_true_estimate} and following the same procedure as for the consistency of $\hat{A}_1$, it can be seen that:
\begin{equation}
\hat{A}_2 \xrightarrow{a.s.} A_2^0 \quad \textmd{and} \quad \hat{B}_2 \xrightarrow{a.s.} B_2^0 \textmd{ as } \min\{M, N\} \rightarrow \infty.
\end{equation} 
It is evident that the result can be easily extended for any $2 \leqslant k \leqslant p$.
\item  \underline{Case 2:} When $k = p+1$:
\begin{equation}\begin{split}\label{A_p+1_B_p+1_estimates}
\begin{bmatrix}
\hat{A}_{p+1} \\ \hat{B}_{p+1}
\end{bmatrix} = \begin{bmatrix}
\frac{2}{MN}\sum\limits_{n=1}^{N}\sum\limits_{m=1}^{M}y_p(m,n)\cos(\hat{\alpha}_{p+1} m + \hat{\beta}_{p+1} m^2 + \hat{\gamma}_{p+1} n + \hat{\delta}_{p+1} n^2) + o(1)\\
\frac{2}{MN}\sum\limits_{n=1}^{N}\sum\limits_{m=1}^{M}y_p(m,n)\sin(\hat{\alpha}_{p+1} m + \hat{\beta}_{p+1} m^2 + \hat{\gamma}_{p+1} n + \hat{\delta}_{p+1} n^2) + o(1)
\end{bmatrix},
\end{split}\end{equation}
where 
\begin{equation*}\begin{split}
y_p(m,n) & = y(m,n) - \sum\limits_{j=1}^{p}\bigg\{\hat{A}_j \cos(\hat{\alpha}_j m + \hat{\beta}_j m^2 + \hat{\gamma}_j n + \hat{\delta}_j n^2) + \hat{B}_j \sin(\hat{\alpha}_j m + \hat{\beta}_j m^2 + \hat{\gamma}_j n + \hat{\delta}_j n^2) \bigg\} \\
& = X(m,n) + o(1), \textmd{ using \eqref{relationship_first_comp_true_estimate} and case 1 results.}
\end{split}\end{equation*}
From here, it is not difficult to see that \eqref{A_p+1_B_p+1_estimates} implies the following result:
\begin{equation*}
\hat{A}_{p+1} \xrightarrow{a.s.} 0 \quad \textmd{and} \quad \hat{B}_{p+1} \xrightarrow{a.s.} 0.
\end{equation*}
This is obtained using Lemma 2 of Lahiri et al. \cite{2015}. It is apparent that the result holds true for any $k > p.$
\end{itemize}\qed \\

\noindent  \textbf{Proof of Theorem \ref{theorem:asymp_dist_multiple_comp_LSE_alphak_betak}:} Consider  \eqref{taylor_series_R_1_MN'_original} and multiply both sides of the equation with the diagonal matrix, $\textit{\textbf{D}}_1^{-1}$:
\begin{equation}\label{taylor_series_R_1_MN'_D)}
(\hat{\boldsymbol{\xi}}_1 - \boldsymbol{\xi}_1^0)\textit{\textbf{D}}_1^{-1} = - \textit{\textbf{R}}_{1, MN}^{(1)'}(\boldsymbol{\xi}_1^0)\textit{\textbf{D}}_1[\textit{\textbf{D}}_1\textit{\textbf{R}}_{1, MN}^{(1)''}(\bar{\boldsymbol{\xi}}_1)\textit{\textbf{D}}_1]^{-1}.
\end{equation}
Computing the elements of the first derivative vector $- \textit{\textbf{R}}_{1, MN}^{(1)'}(\boldsymbol{\xi}_1^0)\textit{\textbf{D}}_1$ and using definition \eqref{definition_R_1_MN} and the preliminary result \eqref{prelim_LSE_first_derivative_multiple_comp_2} (Section \ref{sec:multiple_comp_1D}), we obtain the following result:
\begin{equation}
- \textit{\textbf{R}}_{1, MN}^{(1)'}(\boldsymbol{\xi}_1^0)\textit{\textbf{D}}_1 \xrightarrow{d} \boldsymbol{\mathcal{N}}_2(\textbf{0}, 2\sigma^2 \boldsymbol{\Sigma}_1^{-1}) \textmd{ as } M \rightarrow \infty. \label{distribution_R_1_MN'}
\end{equation}
On combining \eqref{taylor_series_R_1_MN'_D)}, \eqref{distribution_R_1_MN'} and \eqref{limit_R_1_MN''}, we have:
\begin{equation*}
(\hat{\boldsymbol{\xi}}_1 - \boldsymbol{\xi}_1^0)\textit{\textbf{D}}_1^{-1} \xrightarrow{d} \boldsymbol{\mathcal{N}}_2(\textbf{0}, 2\sigma^2 \boldsymbol{\Sigma}_1)
\end{equation*}
This result can be extended for $k = 2$ using the relation \eqref{relationship_first_comp_true_estimate} and following the same argument as above. Similarly, we can continue to extend the result for any $k \leqslant p$.

\end{document}